\documentclass[12pt]{article}

\RequirePackage[OT1]{fontenc}
\RequirePackage{amsthm,amsmath,amssymb,amsfonts,graphicx,subfigure,bm,bbm,fullpage,setspace,color}
\RequirePackage[colorlinks,citecolor=blue,urlcolor=blue]{hyperref}
\theoremstyle{plain}

                          \newtheorem{cor}{Corollary}     
\theoremstyle{remark}         
\theoremstyle{definition} 

\newtheorem{theorem}{Theorem}

\newcommand{\vect}[1]{\mbox{\boldmath $ #1$}}
\newcommand{\iid}{\stackrel{\mathrm{iid}}{\sim}}
\newcommand{\ind}{\stackrel{\mathrm{ind}}{\sim}}

\newcommand\real{\mathbb{R}}

\newcommand\T{\mathcal{T}}

\newcommand\bS{\bm{S}}

\newcommand\by{\bm{y}}
\newcommand\bz{\bm{z}}

\newcommand\bg{\bm{g}}

\newcommand\bups{\bm{\upsilon}}

\newcommand\brho{\bm\rho}

\newcommand\bbet{\bm\beta}

\newcommand\bthe{\bm\theta}
\newcommand\bkap{\bm\kappa}

\newcommand\D{\mathcal{D}}

    \def\independenT#1#2{\mathrel{\setbox0\hbox{$#1#2$}%
    \copy0\kern-\wd0\mkern4mu\box0}}

\newcommand{\beginsupplement}{%
        \setcounter{table}{0}
        \renewcommand{\thetable}{S\arabic{table}}%
        \setcounter{figure}{0}
        \renewcommand{\thefigure}{S\arabic{figure}}%
 }

\usepackage{varioref}		
\labelformat{section}{Section~#1}
\labelformat{figure}{Figure~#1}
\labelformat{table}{Table~#1}

\title{Efficient functional ANOVA through wavelet-domain Markov groves}
  \author{Li Ma \hspace{5em} Jacopo Soriano \\
    \hspace{-2em} Duke University \hspace{2.5em}  Duke University
}

\begin{document}

\maketitle
\doublespacing 

\begin{abstract}
We introduce a wavelet-domain functional analysis of variance (fANOVA) method based on a Bayesian hierarchical model. The factor effects are modeled through a spike-and-slab mixture at each location-scale combination along with a normal-inverse-Gamma (NIG) conjugate setup for the coefficients and errors. A graphical model called the Markov grove (MG) is designed to jointly model the spike-and-slab statuses at all location-scale combinations, which incorporates the clustering of each factor effect in the wavelet-domain thereby allowing borrowing of strength across location and scale. The posterior of this NIG-MG model is analytically available through a pyramid algorithm of the same computational complexity as Mallat's pyramid algorithm for discrete wavelet transform, i.e., linear in both the number of observations and the number of locations. Posterior probabilities of factor contributions can also be computed through pyramid recursion, and exact samples from the posterior can be drawn without MCMC. We investigate the performance of our method through extensive simulation and show that it outperforms existing wavelet-domain fANOVA methods in a variety of common settings. We apply the method to analyzing the orthosis data.
\end{abstract}

\section{Introduction}
\vspace{-0.5em}

This work concerns a common inference task---identifying the contributions from various sources to the variation in functional data, or functional analysis of variance (fANOVA) \cite{ramsay&silverman2005fda,zhang:2014}. Suppose functional observations are measured at a number of given locations. (Throughout this work we use ``locations'' in a general sense to refer to the points in the index or coordinate space on which the function is observed. For example, in time-series applications, these ``locations'' are time points.) A simple approach to fANOVA is to carry out ANOVA---e.g., through an $F$-test---at each location. 
Doing fANOVA in this location-by-location manner often results in poor performance due to the limited amount of data available at each location as well as the necessary multiple testing correction incurred.

An alternative, often more effective approach, is to first apply a basis transformation to the original observations and then carry out ANOVA under the new basis. Many common bases for functional data analysis can be adopted, including splines, polynomials, and Fourier basis, etc.\ \cite{ramsay&silverman2005fda}. With a properly chosen basis, this can allow more effective pooling of information across multiple locations, as well as more sparse representations of the underlying structures (here the functional variations), thereby enhancing the ability to identify such structures. The effectiveness of the different bases depends on the nature of the data; no basis is universally the best for all problems. Here our attention is focused on the wavelet basis transform, and in particular the discrete wavelet transform (DWT) \cite{mallat:1989}.

The DWT is a basis transform for functional data observed on an equi-spaced grid of locations. It has been extensively applied in applications such as signal processing and time-series analysis \cite{percival2006wavelet}. The transformed data, in the form of the so-called wavelet coefficients, characterize functional features of different scales (also referred to as frequencies) at different locations. Each wavelet coefficient is associated with one location-scale (also referred to as time-frequency) combination. This wavelet-domain representation of functional data enjoys several desirable properties. First, the wavelet transform has a
``whitening effect'' that reduces the correlation in the noise, making the common assumption of independent errors more reasonable than in the original space. A second benefit is that
wavelet transforms concentrate ``energy''---in an information theoretic sense as measured by entropy or Kullback-Leibler divergence---into a small number of location-scale combinations, and thus ``sparsify'' the underlying signal thereby making detecting such structures much easier. These properties have motivated the development of numerous wavelet-based regression methods in function estimation. Such methods are particularly effective in comparison to other functional methods when the underlying functions contain local structures~\cite{morris2006wavelet}.

More recently, several authors have proposed methods for fANOVA under the DWT \cite{rosner2000wavelet,abramovich2004optimal,abramovich2006testing,antoniadis2007estimation,mckay:2013}. 
In \cite{rosner2000wavelet,mckay:2013}, the authors treat each location-scale combination individually and carry out an ANOVA test for each, and then identify significant $p$-values after correcting for multiple testing. 
In a similar vein but instead of taking the coefficient-by-coefficient testing approach, \cite{abramovich2006testing} and \cite{abramovich2004optimal}  consider testing of the joint null hypothesis of nil factor effects over all location-scale combinations together, and constructed tests for this purpose that are minimax optimal. On the other hand, while not directly addressing fANOVA, \cite{morris2006wavelet} proposes a Bayesian functional mixed-effects model in the wavelet domain that can be applied to this problem. In particular, a regression model is adopted for the wavelet coefficient at each location-scale combination. ANOVA hypothesis testing can then be naturally handled as a model selection problem using the spike-and-slab prior on the inclusion of the fixed/random effects into the regression model. 
Inference under the model incurs a heavy computational cost requiring Markov Chain Monte Carlo (MCMC) on a large number of regression models, one for each location-scale combination. 
Instead, \cite{antoniadis2007estimation} introduces a frequentist approach for the wavelet-domain mixed effects models addressing testing in both random and fixed effects.

A key motivation for the current work is an important common phenomenon in wavelet-domain analysis that has not been amply exploited in the existing wavelet-based fANOVA methods---the location-scale clustering of functional features in the wavelet coefficients. That is, interesting functional features tend to appear in clusters (like a string of grapes) in the wavelet domain. Such correlation structure has been noted as early as \cite{donoho1994ideal}, and has been fruitfully exploited by \cite{crouse1998wavelet} in function estimation.
It is easy to imagine, and will be confirmed herein, that the location-scale dependency is prevalent in fANOVA problems as well: when a factor impacts the variance of a wavelet coefficient at one location-scale combination, it typically contributes to the variance at nearby locations and in adjacent scales as well. Thus inference should benefit, in fact substantially, from ``borrowing strength'' among nearby and/or nested location-scale combinations in identifying the variance components.

The Bayesian wavelet regression framework for fANOVA affords a natural, principled way to incorporating such dependency through designing model space priors---that is, priors on the spike-and-slab indicators that encode whether a factor contributes to the variance at different location-scale combinations. In particular, we present a model space prior in the form of a graphical model consisting of a collection of Markov trees \cite{crouse1998wavelet}, one for each factor, and hence called the Markov grove (MG). The MG prior is highly parsimonious---specified by a small number of hyperparameters,  and yet flexible enough to characterize the key dependency pattern in factor effects across adjacent/nested location-scale combinations.

Our new Bayesian hierarchical fANOVA model enjoys several important properties. First, due to the tree structure of the MG prior, when coupled with a normal-inverse-Gamma (NIG) conjugate prior specification on the regression coefficients and error variance, exact Bayesian inference for fANOVA can be achieved efficiently. In particular, we show that the joint posterior of our model has a closed form representation computable using a pyramid algorithm that operationally imitates Mallat's pyramid algorithm for the DWT \cite{mallat:1989} and achieves the same computational complexity (or simplicity rather), being linear in both the number of functional observations and the number of locations. The closed form posterior allows direct sampling from the posterior using standard Monte Carlo as opposed to MCMC. Furthermore, when testing fANOVA hypotheses, the posterior marginal probability for the alternative hypotheses (i.e., the presence of factor effects) can also be computed analytically using pyramid recursion without Monte Carlo. This makes our model particularly favorable in large-scale problems such as genomics where fANOVA needs to be completed many times.

The rest of the paper is organized as follows. In \ref{sec:method} we present our methodology.
First, we provide a brief background on Bayesian wavelet regression in Section \ref{sec:single_functional_hmt}. There we review the NIG conjugate prior and show how to use it in conjunction with the MT model to achieve adaptive shrinkage in the wavelet coefficients in a way that takes into account location-scale dependency. In Section~\ref{sec:one_way_fanova} we introduce the MG model as a generalization to the MT model, and show how to use it with the NIG specification to form a hierarchical model for wavelet-based fANOVA. We construct a full inference framework for fANOVA under this model consisting of (i) a closed form of the joint posterior computable through a pyramid algorithm, (ii) a recipe for evaluating the posterior marginal alternative probability of each factor effect at each location-scale combination based on another pyramid algorithm, and (iii) a decision rule for calling significant factor effects that properly adjusts for multiple testing. 
In \ref{sec:numerical_examples_fanova} we carry out simulations to evaluate the performance of our method and compare it to a number of wavelet-domain fANOVA methods. We also apply our method to the analysis of the orthosis data. We conclude in \ref{sec:conclusion} with brief remarks.

\vspace{-1em}

\section{Method}
\label{sec:method}
\vspace{-0.7em}

\subsection{Wavelet regression with normal-inverse-Gamma Markov tree}\label{sec:single_functional_hmt}
\vspace{-0.5em}

We start from considering Bayesian modeling of a single functional observation in the wavelet domain. We shall use this simpler problem as a medium to introduce a number of  building blocks of our more general wavelet-based fANOVA method---namely, (i) Bayesian adaptive wavelet shrinkage with the spike-and-slab prior, (ii) the normal-inverse-Gamma (NIG) conjugate specification, and (iii) the Markov tree (MT) model. We will show how one can use these three tools in conjunction to carry out adaptive shrinkage in the wavelet domain. Our approach arises from a recombination of the ideas from \cite{chipman1997adaptive}, \cite{clyde1998multiple}, \cite{crouse1998wavelet}, \cite{brown:2001}, and \cite{morris2006wavelet}.

Suppose we have a single functional observation whose values are attained
at $T$ equidistant locations $\vect{y} = (y_1, \ldots, y_T)$, and 
\vspace{-3em}

\begin{equation}\label{eq:model1}
 \begin{split}
  \vect{y} & = \vect{f} + \vect{\epsilon} \\ 
\vect{\epsilon} & \sim {\rm N}(\vect{0}, \Sigma_{\epsilon}),
 \end{split}
\end{equation}
\vspace{-2em}

\noindent where $\Sigma_{\epsilon} = \text{diag}(\sigma_{1}^2, \ldots, \sigma_T^2)$. In words, the errors are assumed to be independent across the locations but can be heterogeneous. We
wish to recover the unknown function $\vect{f}$ (or some features of it) from the noisy observation $\vect{y}$. For simplicity, we assume that $T=2^{J+1}$ for some integer $J$.
After applying the DWT to $\vect{y}$ we obtain:
\vspace{-1.5em}

$$
\vect{d} = \vect{z} + \vect{u},
$$
\vspace{-2.5em}

\noindent where $\vect{d} = \vect{y} W'$, $\vect{z} = \vect{f} W'$ and $\vect{u} =
\vect{\epsilon} W'$ with $W$ being the orthonormal matrix corresponding to the
corresponding wavelet basis. Due to properties of multivariate Gaussians, $\vect{u}$ is also
a multivariate Gaussian with a diagonal covariance matrix. 

The elements of $\vect{d}\in\real^{T}$ are referred to as the (empirical) wavelet coefficients, and those in $\vect{z}\in \real^{T}$ are the wavelet coefficients of its
mean function $\vect{f}$. In particular, one of the elements in each of $\vect{d}$, $\vect{z}$, and $\vect{u}$ is called the {\em father} (or scaling) coefficient, while the other $T-1$ are the (mother) coefficients. Without loss of generality, in this work we assume that the scaling coefficients are computed at the coarsest level. 
The elements of $\vect{d}, \vect{z}$ and $\vect{u}$
can be organized into a bifurcating tree structure, with each element in the
corresponding vector associated to a node in the tree. We use the pair of
indices $(j,k)$, where $j=0, \ldots, J = \log_2 T - 1$ and $k=0, \ldots, 2^j -
1$, to represent the $k$th node in the $j$th level of this tree. The two
children nodes of node $(j,k)$ are indexed by $(j+1, 2k)$  and $(j+1,2k+1)$.
Correspondingly, for $j>1$, the parent of $(j,k)$ is indexed by $(j-1, \lfloor
k/2 \rfloor)$. From now on we shall use {\em node} and {\em location-scale combination} interchangeably, and use $\mathcal{T}$ to denote the collection of indices $(j,k)$ corresponding to all nodes in the bifurcating tree. We shall use $d_{j,k}$, $z_{j,k}$, and $u_{j,k}$ to denote the corresponding mother wavelet coefficients.  

The model can be written in a node-specific manner (for notational simplicity, we express the model in terms of the mother coefficients, but the same holds for the father coefficients):
\vspace{-3.5em}

$$
d_{j,k} = z_{j,k} + u_{j,k} \quad \text{where} \quad u_{j,k} \sim {\rm N}(0,\sigma_{j,k}^2).
$$
\vspace{-2.3em}

\noindent Bayesian inference on this regression model proceeds by placing priors on $z_{j,k}$ as well as on the hyperparameter $\sigma_{j,k}^2$. 

It is well-known that effective inference in the wavelet regression should exploit the underlying sparsity of the wavelet coefficients---that is, many of the coefficients $z_{j,k}$ are (or very close to) zero---due to energy concentration. Hence, data-adaptive shrinkage toward zero is critical for effectively inferring $z_{j,k}$. A very popular Bayesian strategy to achieving this, which has been adopted by several authors, is to place a two-group, or so-called spike-and-slab, mixture prior on $z_{j,k}$ \cite{chipman1997adaptive,clyde1998multiple,clyde:2000,brown:2001,morris2006wavelet}:
\vspace{-1.5em}

$$
z_{j,k} \sim (1 - \pi_{j,k}) \cdot \delta(0) + \pi_{j,k} \cdot N(0, \tau_j
\sigma_{j,k}^2).
$$
In words, with prior probability $\pi_{j,k}$, the wavelet coefficient $z_{j,k}$ is non-zero and its value is generated from a Gaussian distribution. (More generally, the spike does not have to be exactly at 0 but can be a Gaussian with a much smaller variance, to which our method will also apply.)
The hyperparameter $\tau_j$ is a level-specific dispersion parameter that characterizes the overall level of variability in the wavelet coefficients at level $j$. Specifically, we consider the following parametric structure as proposed in \cite{abramovich1998wavelet}:
\vspace{-1.5em}

$$
\tau_j = 2^{-\alpha j}\tau, 
$$
\vspace{-2.3em}

\noindent for some $\alpha, \tau > 0$. This implies 
that the wavelet coefficients tend to be smaller as the level $j$
increases.  The parameter $\alpha$ controls the
smoothness  of the functional observations. Larger $\alpha$ corresponds to smoother functions. 
In particular, \cite{abramovich1998wavelet} discusses guiding principles for selecting $\alpha$ to generate functions of various regularities. A moderate choice that works well for a variety of common functions, as recommended in \cite{abramovich1998wavelet}, is $\alpha=0.5$. 
Alternatively, $\alpha$ and $\tau$ can both be chosen through an empirical Bayes approach by maximizing the marginal likelihood (see Section~\ref{sec:prior} for details), which as we will illustrate in the numerical examples can often lead to better performance through incorporating additional adaptivity.

The spike-and-slab prior can be written hierarchically with the introduction of a hidden state $S_{j,k} \in \{ 0,1 \}$ \cite{clyde:2000}, with $S_k=1$ indicating that the empirical wavelet coefficient $d_{j,k}$ contains a signal $z_{j,k}\neq 0$, and when $S_k=0$ otherwise. Formally,
\vspace{-3em}

\begin{align}\label{eq:spike_and_slab}
 z_{j,k} | S_{j,k} &\sim 
 {\rm N}(0, S_{j,k}\cdot \tau_j \sigma_{j,k}^2)
\end{align}
\vspace{-3em}

The error variance $\sigma_{j,k}^2$ is typically unknown, and it can be inferred from the data. In many applications of wavelet regression, the error variance is assumed to be homogeneous, i.e., $\sigma_{j,k}^2\equiv \sigma_0^2$ for all $j$ and $k$.
It has been noted that homogeneous error variance is often unrealistic \cite{morris2006wavelet}. Thus we allow $\sigma_{j,k}^2$ to be heterogeneous and adopt a hyperprior on them:
\vspace{-1.7em}

\begin{equation}\label{eq:inverse_gamma_prior}
\sigma_{j,k}^2 \ind \text{Inv-Gamma}(\nu + 1, \nu \sigma_0^2).
\end{equation}
\vspace{-2.5em}

\noindent The inverse-Gamma prior maintains conjugacy to
the Gaussian model, and consequently the marginal likelihood can be evaluated
analytically. This hierarchical specification includes the homogeneous variance as a special case because as $\nu \to  \infty$, $\sigma_{j,k}^2 \overset{p}{\to} \sigma_0^2$.

Donoho and Johnstone \cite{donoho1994ideal} noted a prevalent phenomenon in many applications of
inference in wavelet spaces: the ``signals''---the wavelet coefficients 
that are large
in magnitude---often show up in clusters in the location-scale tree.
This phenomenon, for instance, can be clearly seen in the 
 four test functions presented in  \cite{donoho1994ideal}
(see \ref{fig:donoho_functions}). 
In particular, when the coefficient $z_{j,k}$ deviates far away from zero, the coefficients of the two children in the bifurcating location-scale tree, namely $z_{j+1,2k}$ and $z_{j+1,2k+1}$ tend to be away from zero as well. Such location-scale clustering is particularly strong for functions with sharp boundaries and abrupt changes such as blocks, bumps, and doppler. Crouse {\it et al} \cite{crouse1998wavelet} pointed out that the clustering pattern in the wavelet coefficients can be directly exploited to improve adaptive shrinkage, and proposed a graphical modeling strategy to induce such spatial-scale dependency by jointly modeling the latent
states $S_{j,k}$ using a Markov process, resulting in a hidden Markov model
evolving on the location-scale tree, called the \emph{Markov tree} (MT).

\begin{figure}[t]
 \centering
 \includegraphics[width= 0.7\textwidth]{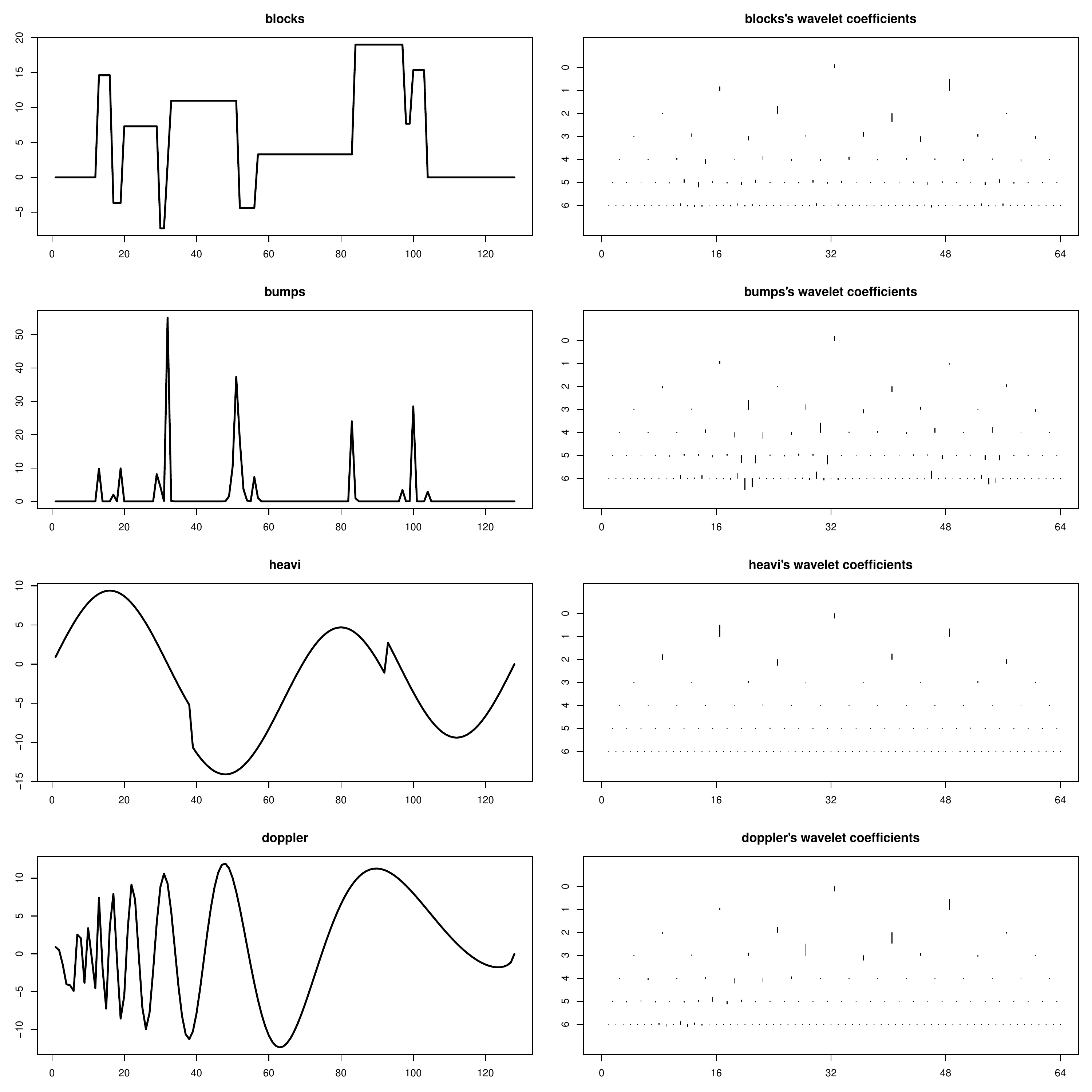}
 \caption{The four test functions  from \cite{donoho1994ideal} and the
associated mother wavelet coefficients. The coefficients for the test
functions \emph{blocks}, \emph{bumps} and \emph{doppler} display strong patterns of location-scale clustering.}
\label{fig:donoho_functions}
\end{figure}

Under the MT for $\bS=\{S_{j,k}:(j,k)\in\T\}$, the shrinkage state of node $(j,k)$ depends on that of its 
parent through a Markov transition:
\vspace{-1.7em}

\begin{equation}\label{eq:transition_probs}
\Pr(S_{j,k} = s' | S_{j-1,\lfloor
k/2 \rfloor } = s) = \rho_{j,k}(s,s'),
\end{equation}
\vspace{-2.3em}

\noindent where $\rho_{j,k}(s,s')$ for $s,s'\in\{0,1\}$ are called the {\em state transition probabilities}, and they can be organized into a $2\times 2$ transition matrix, $\brho_{j,k}$ for each node $(j,k)$. 

A simple and flexible two-hyperparameter specification of these transition matrices is:
\begin{equation}\label{eq:trans_matrix_wave}
\vect{\rho}_{j,k} = \left[ 
\begin{array}{ll}
 \rho_{j,k}(0,0) & \rho_{j,k}(0,1) \\
 \rho_{j,k}(1,0) & \rho_{j,k}(1,1)
\end{array}
\right] = \left[ 
\begin{array}{ll}
 \max\{1 - \eta 2^{-j},0\} & \min\{\eta 2^{-j},1\} \\
 1 - \gamma & \gamma
\end{array}
\right],
\end{equation}
where $0< \gamma < 1$ 
and $\eta > 0$. 
The parameter $\gamma$ induces the spatial-scale  dependency of the wavelet signal.
Larger $\gamma$ values correspond to stronger correlation or clustering in the large wavelet coefficients.
On the other hand, the parameter $\eta$ controls how likely it is to have a ``signal'', i.e., non-zero
wavelet coefficient in each level. 
The exponential decaying factor $2^{-j}$ counters exactly the exponential increase in the expected number of wavelet coefficients in higher resolution, and keeps the prior expected number of {\em de novo} signals (in the sense that a node contains a signal but its parent does not) in each resolution fixed at $\eta$. 

There are two strategies to choosing the hyperparameters $(\eta,\gamma)$. One is to elicit them based on some criteria for multiplicity adjustment and the other is to choose them
using an empirical Bayes approach by maximizing the marginal likelihood. We shall discuss both strategies in Section~\ref{sec:prior}.

Because the root
node $(0,0)$ does not have a parent, the initial state of the process, $S_{0,0}$,
is specified by a set of initial state probabilities $\vect{\rho}_{0,0} =
(\rho_{0,0}(0), \rho_{0,0}(1))$ such that: 
\vspace{-1.5em}

\begin{equation}\label{eq:initial_conditions} 
\Pr(S_{0,0}=s) = \rho_{0,0}(s) \quad \text{for } s \in \{0,1\}.
\end{equation}
\vspace{-2em}

\noindent Combining the MT model \cite{crouse1998wavelet} on $\bS$ and the NIG hierarchical setup \cite{clyde1998multiple, clyde:2000}, Eqs.~\eqref{eq:spike_and_slab},
\eqref{eq:inverse_gamma_prior},
\eqref{eq:transition_probs}, and \eqref{eq:initial_conditions}  
together give a new hierarchical model for the wavelet coefficients, which we shall refer to as the {\em normal inverse-Gamma Markov tree},
or NIG-MT.

We next show how to do inference under the NIG-MT model.
In particular, we show that the tree nature of the MT combined with the
normal inverse-Gamma setup results in full conjugacy of the NIG-MT: the joint posterior on $\{ S_{j,k},
z_{j,k}, \sigma_{j,k}^2 : (j,k) \in \mathcal{T} \}$ is still an NIG-MT whose parameters can be computed
analytically, and can be sampled from directly. 

To this end, let us consider a more general case with $n(\geq 1)$ i.i.d.\ functional
observations $\vect{y}^{(1)}, \vect{y}^{(2)}, \ldots, \vect{y}^{(n)}$ from 
model  \eqref{eq:model1}. From now on, we shall use the superscript ``$(i)$''
to indicate the terms corresponding to the $i$th observation. The node-specific
model after DWT becomes:
\vspace{-3em}

$$
d_{j,k}^{(i)} = z_{j,k} + u_{j,k}^{(i)} \quad \text{where} \quad u_{j,k}^{(i)}
\sim N(0, \sigma_{j,k}^2).
$$
\vspace{-2em}

\noindent Our interest lies in finding the posterior distribution on $\{ S_{j,k},
z_{j,k}, \sigma_{j,k}^2 : (j,k) \in \mathcal{T} \}$ given the observed data. Let $m_{j,k}(s)$ be the marginal likelihood for the node-specific
model on $(j,k)$ given that $S_{j,k} = s \in \{0,1\}$:
\vspace{-1.2em}

$$
m_{j,k}(s)  = \int p(\vect{d}^{(1)}, \ldots, \vect{d}^{(n)} | S_{j,k}=s,
z_{j,k},  \sigma_{j,k}^2)\, \pi(z_{j,k}, \sigma_{j,k}^2)\, d z_{j,k} d
\sigma_{j,k}^2.
$$
From the normal-inverse-Gamma
conjugacy, the marginal likelihood is in closed form:
$$
m_{j,k}(s) = 
   \frac{( \nu \sigma_0^2 )^{\nu + 1} \Gamma(\nu + n/2+1
)}{(2\pi)^{n/2}\Gamma(\nu + 1)} \cdot \bigg[ \dfrac{\tau_j^{-1}}{ n + \tau_j^{-1} } \bigg]^{s/2} \cdot
\bigg[ \nu
\sigma_0^2 + \dfrac{1}{2}\bigg( 
\sum_i (d_{j,k}^{(i)})^2  - 
s\cdot \dfrac{( n \bar{d}_{j,k}  )^2}{ n+ \tau_j^{-1} }
\bigg)
\bigg]^{ - \nu - n/2 - 1}
$$
where 
$\bar{d}_{j,k} = \sum_i d_{j,k}^{(i)} / n$.

The following theorem shows that the NIG-MT model is completely {\em conjugate} in the sense that the joint posterior is still an NIG-MT. Moreover, the posterior hyperparamters are  available analytically through a recursive algorithm operationally similar to the pyramid algorithm for DWT \cite{mallat:1989}. From now on, we shall use $\mathcal{D}$ to represent the totality of
data.

\begin{theorem}\label{thm:post_hmt}
 The joint posterior on $\{ S_{j,k},
z_{j,k}, \sigma_{j,k}^2 : (j,k) \in \mathcal{T} \}$ is still an NIG-MT as follows:
\begin{itemize}
 \item The posterior of the hidden states $S_{j,k}$ is still a MT:
\begin{enumerate}
 \item State transition probabilities:
 \vspace{-1em}
 
 $$
  \Pr(S_{j,k} = s' | S_{j-1, \lfloor
k/2 \rfloor} = s, \mathcal{D} ) = \rho_{j,k}(s,s') \dfrac{ \phi_{j,k}(s') }{
\xi_{j,k}(s)}, 
 $$ 
 \vspace{-2em}
 
 \noindent for $s,s' \in \{0,1\}$ and $ j=1,2, \ldots, J$;
 \item Initial state probabilities:
  $$
  \Pr(S_{0,0} = s | \mathcal{D} ) = \rho_{0,0}(s) \dfrac{ \phi_{0,0}(s) }{
\xi_{0,0}(0)} \quad \text{for } s \in \{0,1\}.
 $$ 
\end{enumerate}
 \item The posterior of the variances $\sigma_{j,k}^2$ given $S_{j,k}$ is:
$$
[ \sigma_{j,k}^2 | S_{j,k}, \mathcal{D} ] \sim \text{Inv-Gamma}
\bigg(\nu+1 + \dfrac{n}{2},  \nu
\sigma_0^2 + \dfrac{1}{2}\bigg( 
\sum_i (d_{j,k}^{(i)})^2  - 
\dfrac{S_{j,k}\cdot ( n \bar{d}_{j,k}  )^2}{ n+ \tau_j^{-1} }
\bigg)
\bigg).
$$
 \item The posterior of $z_{j,k}$ given $S_{j,k}$ and $\sigma_{j,k}^2$ is:
$$
[z_{j,k} | \sigma_{j,k}^2, S_{j,k}, \mathcal{D}] \sim
{\rm N}\bigg( \dfrac{S_{j,k}\cdot n\bar{d}_{j,k}}{n+\tau_j^{-1}},
\dfrac{S_{j,k} \cdot \sigma_{j,k}^2}{n+\tau_j^{-1}}  \bigg). $$
\end{itemize}
The mappings $\phi_{j,k}$ and $\xi_{j,k}: \{0,1\} \mapsto [0, + \infty)$
are defined recursively in $j$ and can be computed through a bottom-up pyramid algorithm as follows:
\begin{align*}
\phi_{j,k}(s)& =
\left\{ 
\begin{array}{ll}
 m_{j,k}(s) \cdot \xi_{j+1,2k}(s) \cdot \xi_{j+1,2k+1}(s) & \text{for
$j=0,1,2\ldots,J-1$}\\
m_{j,k}(s) & \text{for $j=J$,}
\end{array}  \right. \\
 \xi_{j,k}(s) & = 
\left\{
\begin{array}{ll}
 \sum_{s'\in\{0,1\}}
\rho_{j,k}(s,s')\cdot \phi_{j,k}(s') & \text{for
$j=1,2,\ldots,J$}\\
\sum_{s'\in\{0,1\}} \rho_{0,0}(s')\cdot \phi_{j,k}(s') & \text{for $j=0$.}
\end{array}
\right.
\end{align*}
\end{theorem}
\noindent Remark I: The recursive computation of the mappings $\phi_{j,k}$ and $\xi_{j,k}$
is operationally analogous to Mallat's pyramid algorithm \cite{mallat:1989} for carrying out the DWT. In the order $J,J-1,\ldots,1,0$, it computes the mapping at a node based on the mapping values on its children in the next resolution. The algorithm achieves the same computational complexity for evaluating the posterior exactly as Mallat's algorithm, that is, linear both in $n$ and in $T$. 
\vspace{0.5em}

\noindent Remark II: The term $\xi_{0,0}(0)$ is the overall marginal likelihood (integrating out all the latent variables) given the hyperparameters, which we can use to set the hyperparameters through a common empirical Bayes strategy---maximum marginal likelihood estimation (MMLE).
\vspace{0.5em}

Given the analytical form of the joint posterior, the posterior mean of $\bz$ can also be computed exactly. To this end, we first use a top-down pyramid algorithm to compute the posterior marginal probability of
the hidden states as follows. In the order $j=1,2,\ldots,J$, the posterior marginal probability of $S_{j,k}$ for each $k$ is available as
\vspace{-2.5em}

\[
\Pr(S_{j,k}= s'\,|\,\mathcal{D})=\sum_{s\in \{0,1\}} \Pr(S_{j-1,\lfloor k/2
\rfloor}=s\,|\,\mathcal{D})\cdot \Pr(S_{j,k}=s'\,|\,S_{j-1,\lfloor k/2
\rfloor}=s,\mathcal{D}).
\]
\vspace{-2em}

\noindent Then the posterior mean of $\vect{z}$ is given by
\vspace{-2.5em}

\begin{align}
\label{eq:z_shrinkage}
\tilde{z}_{j,k}:={\rm E}(z_{j,k}\,|\,\mathcal{D}) = 
\Pr(S_{j,k}=1\,|\,\mathcal{D})\cdot \dfrac{n}{n+\tau_j^{-1}}  \bar{d}_{j,k},
\end{align}
\vspace{-2.5em}

\noindent which has an intuitive explanation in terms of shrinkage. The average of the
observed wavelet coefficients $\bar{d}_{j,k}$ is shrunk toward the prior mean 0
with the amount of shrinkage being averaged over the different shrinkage states.
By applying an inverse DWT to $\tilde{\vect{z}}$ we can get the posterior mean
of $\vect{f}$, $E(\vect{f}\,|\,\mathcal{D}) = W^{-1}\tilde{\vect{z}}$.

In addition to the posterior mean, one can construct credible intervals for $\vect{z}$ and for $\vect{f}$ by sampling from the joint posterior of $\{(S_{j,k},z_{j,k},\sigma^2_{j,k}):j=0,2,\ldots,J,k=0,1,\ldots,2^{j}-1\}$ according to Theorem~\ref{thm:post_hmt}. Because the exact posterior is available, no MCMC is needed and the sampling is standard Monte Carlo. Given a posterior samples Bayesian inference can proceed as usual. For example, a credible band for $\vect{f}$ is available from a posterior sample $\vect{f}$, attained through applying an inverse DWT to the posterior sample on~$\vect{z}$. We illustrate the work of this model in function denoising through simulations in Section~\ref{sec:denoising}.

\vspace{-0.7em}

\subsection{Wavelet fANOVA with normal-inverse-Gamma Markov grove}
\label{sec:one_way_fanova}

Next we present our main methodology for fANOVA. We first introduce our framework in one-way fANOVA as the notation is much simpler with all the essential components of the framework present, and then generalize the formulation to the general multi-way case.
\vspace{0.5em}

{\it One-way fANOVA.} Suppose we have $G$ groups of independent functional observations whose values are
attained at $T$ equidistant points $\vect{y}^{(g,i)} = (y_1^{(g,i)}, \ldots,
y_{T}^{(g,i)})$. Suppose
\vspace{-2.5em}

\begin{equation}\label{eq:oneway_model}
 \begin{split}
  \vect{y}^{(g,i)} & = \vect{f}^{(g)} + \vect{\epsilon}^{(g,i)} \\ 
\vect{\epsilon}^{(g,i)} & \sim {\rm N}(\vect{0}, \Sigma_{\epsilon}),
 \end{split}
\end{equation}
\vspace{-2em}

\noindent  where $g=1, \ldots, G$ is the group index, $i = 1, \ldots, n_g$ is the index for the replicates in the $g$th group,  and $\Sigma_{\epsilon} =
\text{diag}(\sigma_{1}^2, \ldots, \sigma_T^2)$. The (one-way) fANOVA problem concerns identifying the variation among $\vect{f}^{(g)}$, if any, from that in the noise.

Treating the first group as the baseline $\vect{f}=\vect{f}^{(1)}$, and letting $\vect{b}^{(g)} = \vect{f}^{(g)}-\vect{f}^{(1)}$ be the contrast of each group to the baseline, we can write
\vspace{-3em}

\begin{align}
\label{eq:oneway_model2}
  \vect{y}^{(g,i)} = \vect{f} + \vect{b}^{(g)} + \vect{\epsilon}^{(g,i)}.
\end{align}
\vspace{-3em}

\noindent By design $\vect{b}^{(1)}=\vect{0}$, and the ANOVA problem boils down to inference on the contrast functions $\vect{b}^{(g)}$ for $g\geq 2$.
After applying the DWT we obtain 
\vspace{-1.5em}

$$
\vect{d}^{(g,i)}  = 
  \vect{z} + \bbet^{(g)} + \vect{u}^{(g,i)} 
$$
\vspace{-2.5em}

\noindent where $\bz=W\vect{f}$, $\bbet^{(g)} = W \vect{b}^{(g)}$, and $\vect{u}^{(g,i)}=W\vect{\epsilon}^{(g,i)}$ for $g=1,2,\ldots,G$, and $i=1,2,\ldots,n_{g}$. We can again write the model in a node-specific manner:
\vspace{-1.5em}

$$
d_{j,k}^{(g,i)} = 
 z_{j,k} + \beta_{j,k}^{(g)} + u_{j,k}^{(g,i)} \quad \text{where } u_{j,k}^{(g,i)} \ind {\rm N}(0,\sigma_{j,k}^2).
$$
\vspace{-2em}

As before, we introduce a latent indicator $S_{j,k}$ for each $z_{j,k}$ and adopt the same NIG-MT setup for $\{z_{j,k},S_{j,k},\sigma^2_{j,k}: (j,k)\in \T\}$ as given in Eqs.~\eqref{eq:spike_and_slab},
\eqref{eq:inverse_gamma_prior}, \eqref{eq:transition_probs}, and \eqref{eq:initial_conditions}. Under this formulation, ANOVA can be accomplished through inference on the contrast coefficients $\beta_{j,k}^{(g)}$. 
To this end, we introduce another latent indicator $R_{j,k}\in\{0,1\}$ such that
\vspace{-2.7em}

\begin{align}
\label{eq:r_spike_and_slab}
[ \beta_{j,k}^{(2)}, \ldots, \beta_{j,k}^{(G)} | \sigma_{j,k}^2, R_{j,k} ] \iid
 {\rm N}(0, R_{j,k}\cdot  \upsilon_j\sigma_{j,k}^2)
\end{align}
\vspace{-2.8em}

\noindent where similar to $\tau_j$, $\upsilon_j = 2^{-\alpha j} \upsilon$ 
is a scaling parameter that characterizes the size of the differences across the factor levels. Our motivation to defining a different scaling parameter $\upsilon$ than $\tau$ for the $\beta$'s is that
the $\beta$'s characterize the difference across the factor levels while the $z$ that for the baseline function mean. 
It is often the case that the scale of the $z$'s are substantially different than that of the $\beta$'s.

Just as $\{S_{j,k}:(j,k)\in\T\}$ are modeled in a correlated manner to capture the spatial-scale dependency in $z_{j,k}$, we do that for the $\{R_{j,k}:(j,k)\in\T\}$ as well. Intuitively, if a factor contributes to the variation at one location-scale combination, then it typically contributes to the variation at the children/neighbor nodes as well.
Again, an MT is a convenient choice for jointly modeling the latent indicators $\{R_{j,k}:(j,k)\in\T\}$. 
We let $\bkap_{j,k}$ denote the corresponding state transition matrix (or the initial probability vector when $(j,k)=(0,0)$), which can be specified in the same way as given in  \eqref{eq:trans_matrix_wave} for $\brho_{j,k}$ . That is for $r,r'\in\{0,1\}$,
\vspace{-2.5em}

\begin{align}
\label{eq:r_tran_mat}
 \Pr(R_{j,k} = r' | R_{j-1,\lfloor
k/2 \rfloor } = r) & = \kappa_{j,k}(r,r') \text{ for $j> 0$} \quad \text{and} \quad \Pr(R_{0,0}=r)= \kappa_{0,0}(r).
\end{align}
\vspace{-2.5em}

Now we arrive at a fully specified joint model on $\{z_{j,k},S_{j,k},\bbet_{j,k},R_{j,k},\sigma^2_{j,k}: (j,k)\in \T\}$ given by Eqs.~\eqref{eq:spike_and_slab},
\eqref{eq:inverse_gamma_prior}, \eqref{eq:transition_probs}, \eqref{eq:initial_conditions}, \eqref{eq:r_spike_and_slab}, and \eqref{eq:r_tran_mat}. 
It is specified by the NIG conjugate priors on $(z_{j,k},\bbet_{j,k},\sigma^{2}_{j,k})$ given the latent indicators, and two MTs on the latent indicators. For this reason, we shall refer to this model as a {\em normal-inverse-Gamma Markov grove} (NIG-MG).

Next we show how Bayesian inference can be carried out for the
NIG-MG. It turns out that the joint posterior $\{z_{j,k},
\bbet_{j,k}, S_{j,k}, R_{j,k}, \sigma_{j,k}^2 : (j,k)
\in \mathcal{T} \}$ can again be computed analytically through a
pyramid algorithm whose complexity is linear in both $n$ and $T$. Accordingly, posterior marginal and joint null/alternative probabilities can also be evaluated exactly, and one can  sample from the exact posterior using standard Monte Carlo.

We write the node-specific model in matrix notation:
\vspace{-1.7em}

$$
\vect{d}_{j,k} = X \vect{\theta}_{j,k} + \vect{u}_{j,k},
$$
\vspace{-2.5em}

\noindent where  $\vect{d}_{j,k} = (d_{j,k}^{(1,1)}, \ldots, d_{j,k}^{(n_G,G)})'$
is the vector of the wavelet coefficients for all the observations at node
$(j,k)$, $\vect{\theta}_{j,k} = (z_{j,k},
\beta_{j,k}^{(2)}, \ldots, \beta_{j,k}^{(G)})'$ is the vector of the wavelet
coefficients for the mean functions, $\vect{u}_{j,k} =
(u_{j,k}^{(1,1)}, \ldots, u_{j,k}^{(n_G,G)})' $ is the vector of the residual
errors, and $X$ is the design matrix.  The design matrix can be written as
\vspace{-1.7em}

 $$
X = ( \mathbbm{1}_n, \vect{e}_2, \ldots, \vect{e}_G ),
$$
\vspace{-2.5em}

\noindent where $ n = \sum_{g=1}^G n_g$, $\mathbbm{1}_n$ is a vector of $n$ ones, and $\vect{e}_g$ is a binary vector where the
$h$th element is equal to one if  the $h$th observation belongs to
group $g$, and equal to zero otherwise. We also define the following matrices for $s,r\in\{0,1\}$:
\vspace{-1.5em}

$$
X(s,r) = 
(s\mathbbm{1}_n, r\vect{e}_2, \ldots, r\vect{e}_G ),\quad 
\Lambda_{j} = 
\text{diag}(1/\tau_j,\underbrace{1/\upsilon_j,1/\upsilon_j,\ldots,1/\upsilon_j}_{\text{$G-1$ copies}} ),
$$
and 
\[
M(s,r) =  \left(
\begin{array}{ll}
 s & sr \mathbbm{1}'_{G-1} \\
 sr \mathbbm{1}_{G-1} & r  \mathbbm{1}_{G-1}\mathbbm{1}'_{G-1}
\end{array}
\right).
\]

\noindent The marginal likelihood for the node-specific model on
$(j,k)$ given $S_{j,k} = s$ and $R_{j,k} = r$ is
\begin{equation}\label{eq:marginal_like}
m_{j,k}(s,r) = 
  \frac{( \nu \sigma_0^2 )^{\nu+1} \Gamma(\nu + n/2+1)}{(2\pi)^{n/2}\Gamma(\nu + 1)} \cdot \dfrac{ |\Lambda_{j}|^{1/2} }{|\Lambda_{j}^{*}(s,r)|^{1/2}} 
  \cdot \big[ \nu \sigma_0^2
+ \Upsilon_{j,k}(s,r) \big]^{-\nu-n/2-1}.
\end{equation}
where
\vspace{-3.5em}

\begin{align*}
&\Upsilon_{j,k}(s,r)=\left\{ \vect{d}_{j,k}'\vect{d}_{j,k} 
  - [\vect{\mu}_{j,k}^*(s,r)]'   \Lambda_{j}^{*}(s,r)\vect{\mu}_{j,k}^*(s,r)\right\}/2,\\
\Lambda_{j}^*(s,r) &= X(s,r)'X(s,r) + \Lambda_{j}, \quad \text{and} \quad \vect{\mu}_{j,k}^*(s,r) = [\Lambda_{j}^*(s,r)]^{-1} [
 X(s,r)' \vect{d}_{j,k}  ].
 \end{align*}
\vspace{-2.5em}

\begin{theorem}
\label{thm:post_anova}
 The joint posterior on  $\{z_{j,k},
\bbet_{j,k}, S_{j,k}, R_{j,k}, \sigma_{j,k}^2 : (j,k)
\in \mathcal{T} \}$ is as follows.
\begin{itemize}
 \item The marginal posterior of the hidden states 
$ \{ (S_{j,k}, R_{j,k}) : (j,k)
\in \mathcal{T} \}$ is an MT defined on the product state-space $\{0,1\}\times\{0,1\}$ with 
\begin{enumerate}
 \item State transition probabilities:
 \vspace{-3.5em}
 
\begin{align*}
\hspace{-2em} \Pr(S_{j,k} =  s', R_{j,k} =  r' | S_{j-1,\lfloor k/2 \rfloor} = s, 
R_{j-1,\lfloor k/2 \rfloor} =
r,\mathcal{D})= \rho_{j,k}(s,s') \kappa_{j,k}(r,r') 
\phi_{j,k}(s',r')/\xi_{j,k}(s,r),
\end{align*}
\vspace{-3.5em}
 
\noindent for $j=1, \ldots, J$.
 \item Initial state probabilities:
 \vspace{-2em}
 
$$
\Pr(S_{0,0}=s, R_{0,0}=r | \mathcal{D}) = 
\rho_{0,0}(s) \kappa_{0,0}(r) 
\phi_{0,0}(s,r)/\xi_{0,0}(0,0).
$$
\end{enumerate}
\vspace{-0.5em}

\item The conditional posterior of $\sigma_{j,k}^2$ given $S_{j,k}$ and
$R_{j,k}$ is: 
\vspace{-2em}
 
$$
 [ \sigma_{j,k}^2 | S_{j,k}, R_{j,k}, \mathcal{D}] 
\sim \text{Inv-Gamma}\bigg( \nu + 1 + \dfrac{n}{2}, 
\nu \sigma_0^2
+ \Upsilon_{j,k}(S_{j,k},R_{j,k})    \bigg).
$$

\item The posterior of $z_{j,k},
\beta_{j,k}^{(2)}, \ldots, \beta_{j,k}^{(G)}$ given $S_{j,k}$, $R_{j,k}$ and
$\sigma_{j,k}^2$ is given as follows 
\vspace{-3em}

\begin{align*}
&[ z_{j,k}, \beta_{j,k}^{(2)}, \ldots, \beta_{j,k}^{(G)} | 
  \sigma_{j,k}^2, S_{j,k}, R_{j,k},
 \mathcal{D}]\\ 
& \hspace{5em} \sim {\rm N}\bigg( \vect{\mu}^*_{j,k}(S_{j,k},R_{j,k}), \sigma_{j,k}^2\,M(S_{j,k},R_{j,k}) \circ  [
 \Lambda_{j}^*(S_{j,k},R_{j,k}) ]^{-1} \bigg).
\end{align*}
\vspace{-3em}

\end{itemize}
where $\circ$ represents the Hadamard product, and for any matrix $A$. The mappings $\phi_{j,k}(s,r)$ and $\xi_{j,k}(s,r): \{0,1 \} \times \{ 0, 1 \}
\to [0, + \infty) $ can be computed recursively through a pyramid algorithm as follows:
\vspace{-2em}

\begin{align*}
\phi_{j,k}(s,r)& =
\left\{ 
\begin{array}{ll}
 m_{j,k}(s,r) \cdot \xi_{j+1,2k}(s,r) \cdot \xi_{j+1,2k+1}(s,r) & \text{for
$j=0,1,2\ldots,J-1$}\\
m_{j,k}(s,r) & \text{for $j=J$,}
\end{array}  \right. \\
 \xi_{j,k}(s,r) & = 
\left\{
\begin{array}{ll}
 \sum_{ s', r' }
\rho_{j,k}(s,s') \cdot \kappa_{j,k}(r,r') \phi_{j,k}(s',r')  & \text{for
$j=1,2,\ldots,J$}\\
\sum_{s', r'} \rho_{0,0}(s') \cdot \kappa_{0,0}(r')\cdot
\phi_{0,0}(s',r') & \text{for $j=0$.}
\end{array}
\right.
\end{align*}
\end{theorem}
\vspace{0.5em}

Once the joint posterior is computed following the theorem, Bayesian inference can proceed in the usual manner. In particular, for testing the presence of a variance contribution from the factor, $R_{j,k}$ is an indicator for whether the null hypothesis at location-scale $(j,k)$: 
\vspace{-1.5em}

\[
H_{j,k}: \beta_{j,k}^{(1)}=\beta_{j,k}^{(2)}=\cdots = \beta_{j,k}^{(G)} = 0
\]
\vspace{-2em}

\noindent is false. Thus $P(R_{j,k}=1\,|\,\D)$ represents the posterior probability for the factor to contribute to the variation at location-scale $(j,k)$. For this reason, we shall refer to $P(R_{j,k}=1\,|\,\D)$ as the {\em posterior marginal alternative probability} (PMAP) for location-scale $(j,k)$. 
Next we show how to compute PMAPs through a top-down pyramid algorithm.
\begin{cor}
\label{cor:pmap}
 For $j=1,2,\ldots,J$, the posterior marginal distribution of $(S_{j,k},R_{j,k})$ can be computed recursively as
 \vspace{-3.5em}
 
\begin{align*}
&\Pr(S_{j,k}=s',R_{j,k}=r'\,|\,\D)\\ 
=&\sum_{s,r} \Pr(S_{j-1,\lfloor k/2 \rfloor}\!=\!s,R_{j-1,\lfloor k/2 \rfloor}\!=\!r\,|\,\D)\! \times\! \Pr(S_{j,k} \!=\!  s', R_{j,k}\! =\!  r' | S_{j-1,\lfloor k/2 \rfloor} \!=\! s, 
R_{j-1,\lfloor k/2 \rfloor} =
r,\D)
\end{align*}
\vspace{-3em}

\noindent using the initial and transition probabilities given in Theorem~\ref{thm:post_anova}. Then the PMAPs are given by $\Pr(R_{j,k}=1\,|\,\D)= \sum_{s}\Pr(S_{j,k}=s,R_{j,k}=r\,|\,\D)$. 
\end{cor}

In the next corollary, we show how to compute the posterior probability for the presence of factor effects at {\em any} (i.e., at least one) location-scale combination. This probability, which we refer to as the posterior joint alternative probability (PJAP) can be used for testing the ``global'' null hypothesis that the factor does not contribute to the variation at all.
\begin{cor}
\label{cor:pjap}
For all $(j,k)\in\T$ and $s\in\{0,1\}$, let 
\vspace{-2.7em}

\[ \varphi_{j,k}(s) = P(R_{j',k'} = 0 \text{ for all $(j',k')\in\T(j,k)$} \,|\,S_{j-1,\lfloor k/2 \rfloor} = s,R_{j-1,\lfloor k/2 \rfloor} = 0,\D)\]
\vspace{-2.7em}

\noindent where $\T(j,k)$ denotes the subtree in $\T$ with $(j,k)$ as the root, i.e., $\T(j,k)$ includes $(j,k)$ and all of its descendants in $\T$.
Then we can compute $\varphi_{j,k}(s)$ by the following pyramid algorithm
\vspace{-5em}

\begin{align*}
&\varphi_{j,k}(s) \\
=& \begin{cases}
\sum_{s'}\Pr(S_{j,k} =  s', R_{j,k} =  0 | S_{j-1,\lfloor k/2 \rfloor} = s, R_{j-1,\lfloor k/2 \rfloor}=0,\D) \varphi_{j+1,2k}(s')\varphi_{j+1,2k+1}(s') & \text{for $j<J$}\\
\sum_{s'}\Pr(S_{j,k} =  s', R_{j,k} =  0 | S_{j-1,\lfloor k/2 \rfloor} = s, R_{j-1,\lfloor k/2 \rfloor}=0,\D) & \text{for $j=J$.}
\end{cases}
\end{align*}
The posterior joint null probability (PJNP) is given by
\vspace{-1.5em}

\[
\Pr(R_{j,k}=0 \text{ for all $(j,k)\in\T$}\,|\,\D) = \varphi_{0,0}(0).
\]
\vspace{-2em}

\noindent Accordingly, the posterior joint alternative probability (PJAP) is $1-\varphi_{0,0}(0)$.
\end{cor}
\noindent Remark: Corollary~\ref{cor:pmap} and \ref{cor:pjap} can also be applied to the prior model to get the prior marginal alternative probabilities, which can be used to elicit the prior specification on the hyperparameters. See Section~\ref{sec:prior} for more details.
\vspace{0.5em}

Theorem~\ref{thm:post_anova} allows us to draw posterior samples of $\{z_{j,k},
\bbet_{j,k}, S_{j,k}, R_{j,k}, \sigma_{j,k}^2 : (j,k)\in\T\}$ using standard Monte Carlo (not MCMC). Based on this posterior sample, we can also complete other inference tasks such as computing the posterior mean of $\vect{b}^{(g)}$ and constructing credible bands for $\vect{b}^{(g)}$ which quantifies the posterior uncertainty of the factor contribution to the functional variation. This can be achieved by applying inverse DWT to the posterior draws of $\bbet_{j,k}$'s. We will illustrate this in the numerical examples. 

\vspace{1em}

{\it Multi-way fANOVA.} The NIG-MG model for one-way fANOVA can be naturally extended to the case with multiple factors by specifying one MT for each factor to capture the location-scale clustering of each factor effect. The complication is mainly in the notation. 

Suppose now we have $L$ factors, and the $l$th factor has $G_{l}$ levels. Now suppose for each factor combination $\bg=(g_1,g_2,\ldots,g_L)\in \{1,2,\ldots,G_{1}\}\times \{1,2,\ldots,G_{2}\} \times \cdots \times \{1,2,\ldots,G_{L}\}$, we have $n_{\bg}$ independent functional observations $\by^{(\bg,i)}=(y_1^{\bg,i},y_2^{\bg,i},\ldots,y_{T}^{\bg,i})$ for $i=1,2,\ldots,n_{\bg}$, and suppose each observation arise from the following model
\vspace{-2.5em}

\begin{equation}\label{eq:multiway_model}
 \begin{split}
  \vect{y}^{(\bg,i)} & = \vect{f}_1^{(g_1)} + \vect{f}_2^{(g_2)} + \cdots + \vect{f}_L^{(g_L)} + \vect{\epsilon}^{(\bg,i)} \\ 
\vect{\epsilon}^{(\bg,i)} & \sim {\rm N}(\vect{0}, \Sigma_{\epsilon}),
 \end{split}
\end{equation}
\vspace{-2.2em}

\noindent where $g=1, \ldots, G$ is the group index, $i = 1, \ldots, n_g$ is the index for the replicates in the $g$th group,  and $\Sigma_{\epsilon} =
\text{diag}(\sigma_{1}^2, \ldots, \sigma_T^2)$.

Now let $\vect{f} = \sum_{l=1}^{L}\vect{f}_l^{(1)}$ and $\vect{b}^{(g_l)}_{l} = \vect{f}_{l}^{(g_l)}-\vect{f}_{l}^{(1)}$ for $g_l=1,2,\ldots,G_{l}$. Then
\vspace{-1.3em}

\[
\by^{(\bg,i)} = \vect{f} + \sum_{l=1}^{L} \vect{b}^{(g_l)}_{l} + \vect{\epsilon}^{(\bg,i)}.
\]
\vspace{-2.2em}

\noindent After applying the DWT we obtain:
\vspace{-1.5em}

$$
\vect{d}^{(\bg,i)}  = 
  \vect{z} + \bbet_{l}^{(g_l)} + \vect{u}^{(\bg,i)} 
$$
\vspace{-3em}

\noindent where $\bz=W\vect{f}$, $\bbet_{l}^{(g_l)} = W \vect{b}_l^{(g_l)}$, and $\vect{u}^{(\bg,i)}=W\vect{\epsilon}^{(\bg,i)}$ for $\bg$ and $i=1,2,\ldots,n_{\bg}$. Again the corresponding node-specific model is
\vspace{-1em}

$$
d_{j,k}^{(\bg,i)} = 
 z_{j,k} + \sum_{l=1}^{L}\beta_{l\,j,k}^{(g_l)} + u_{j,k}^{(\bg,i)} \quad \text{where } u_{j,k}^{(\bg,i)} \ind {\rm N}(0,\sigma_{j,k}^2).
$$
Just as in the one-way case, in addition to the latent indicators $S_{j,k}$ for $z_{j,k}$, we introduce an indicator $R_{l\,j,k}$ for each factor effect such that $R_{l\,j,k}=1$ if and only if 
\vspace{-1.8em}

\[
H_{l\,j,k}: \beta_{l\,j,k}^{(1)}=\beta_{l\,j,k}^{(2)}=\cdots = \beta_{l\,j,k}^{(G_l)} = 0 \quad \text{ is false.}
\]
\vspace{-2.5em}

\noindent Because each factor contribution will be correlated across location and scale, we adopt an MT model on each of $\{S_{j,k}:(j,k)\in \T\}$ and $\{R_{l\,j,k}:(j,k)\in\T\}$ for $l=1,2,\ldots,L$, to induce the proper location-scale dependency. This results in a ``grove'' of $L+1$ Markov trees,
\vspace{-3.2em}

\begin{align*}
\{S_{j,k}:(j,k)\in \T\} &\sim {\rm MT}(\brho_{j,k}) \quad \text{and} \quad \{R_{l\,j,k}:(j,k)\in\T\} \sim {\rm MT}(\bkap_{j,k}) \quad  \text{for $l=1,2,\ldots,L$},
\end{align*}
\vspace{-3.0em}

\noindent which are independent given the transition matrices $\brho_{j,k}$ and $\bkap_{j,k}$. We specify $\brho_{j,k}$ and $\bkap_{j,k}$ as in Eq.~\eqref{eq:trans_matrix_wave}. We allow the hyperparameters $(\eta,\gamma)$ to be different for $\bkap_{j,k}$ than for $\brho_{j,k}$. This is necessary because the sparsity (as characterized by $\eta$) as well as the spatial-dependency (as characterized by $\gamma$) can be very different for the baseline mean function than for factor effects. For example, the baseline function may be a smooth function, while a factor effect consists of spiky disturbances to the baseline. We will investigate these scenarios in the numerical studies. Specifically, we let $(\eta_{\rho},\gamma_{\rho})$ and $(\eta_{\kappa},\gamma_{\kappa})$ respectively denote the corresponding hyperparameters for $\brho_{j,k}$ and $\bkap_{j,k}$. We discuss prior specification in Section~\ref{sec:prior}. 

The hyperparameter sharing that we enforce here by specifying the same prior transition matrix $\bkap_{j,k}$ for the factor effects not only helps attain parsimony but is reasonable from a modeling perspective as well. In most applications there is no {\em prior} reason to believe any factor contribution to have a different spatial-scale dependency pattern than any other. In situations where one indeed has reasons to believe that the spatial dependency is different among the factor effects, we can specify a different prior transition matrix for the the factors. 

Finally, we still adopt the NIG specification on $z_{j,k}$, $\beta_{l\,j,k}^{(g_l)}$, and $\sigma^2_{j,k}$ as in Eqs.~\eqref{eq:spike_and_slab},
\eqref{eq:inverse_gamma_prior}, and \eqref{eq:r_spike_and_slab}. In particular, we adopt a different scaling parameter $\upsilon_l$ for each factor effect
\vspace{-1.5em}

\[ [ \beta_{l\,j,k}^{(2)}, \ldots, \beta_{l\,j,k}^{(G_l)} | \sigma_{j,k}^2, R_{l\,j,k} ] \iid
 {\rm N}(0, R_{l\,j,k}\cdot  2^{-\alpha j} \upsilon_{l}\cdot \sigma_{j,k}^2)
\]
\vspace{-2.5em}

\noindent because the effect sizes can often be substantially different among the factors.
We now have a fully specified joint model, which is again an NIG-MG consisting of $(L+1)$ Markov trees. 

Bayesian inference under this model proceeds in exactly the same fashion as that for one-way fANOVA. Specifically, the marginal likelihood of the node-specific model given $S_{j,k}=s$ and $R_{l\,j,k}=r_{l\,j,k}$ for all $l$, denoted as $m_{j,k}(s,r_1,r_2,\ldots,r_l)$, also takes the same form as before with the design matrix incorporating all of the factor information. The joint posterior is given by a variant of Theorem~\ref{thm:post_anova}. (See Supplementary Materials~S2 for details.) Similarly, the PMAPs for each factor, $\Pr(R_{l\,j,k}=1\,|\,\D)$, can be computed analytically through pyramid recursion, as well as the PJAP, $\Pr(R_{l\,j,k}=1 \text{ for some $(j,k)\in\T$})$, following Corollary~\ref{cor:pmap} and Corollary~\ref{cor:pjap}. Samples can be drawn from the joint posterior through standard Monte Carlo. 

\vspace{-1em}

\subsection{Prior specification}
\label{sec:prior}
\vspace{-0.5em}

The NIG-MG is specified by the following hyperparameters: $\bthe=(\alpha,\tau,\bups,\sigma_0,\nu,\eta_{\rho},\gamma_{\rho},\eta_{\kappa},\gamma_{\kappa})$, where $\bups=(\upsilon_1,\upsilon_2,\ldots,\upsilon_{L})$. In determining the choice of these hyparameters, it is important to note their interpretations. In particular, the hyperparameters fall into two categories, respectively called the {\em scaling} parameters and the {\em sparsity} parameters. We discuss their specification in turn. 
\vspace{0.25em}

{\em Scaling hyparameters.} This category consists of $\bthe_{scaling}=(\alpha,\tau,\bups,\sigma_0,\nu)$, and they characterize the scale (or size) of either the mean or the factor contribution relative to the size of the errors. Their proper specification depends on the underlying smoothness of the function means and factor effects as well as the signal to noise ratio. We recommend setting these parameters in a data-adaptive manner through empirical Bayes (described later).
\vspace{0.25em}

{\em Sparsity hyparameters.} This category consists of the parameters that determine the prior distributions of the latent indicator variable $S_{j,k}$ and $R_{l\,j,k}$. That is, the parameters that specify the state transition matrices of the MTs: $\bthe_{sparsity}=(\eta_{\rho},\gamma_{\rho},\eta_{\kappa},\gamma_{\kappa})$. In particular, they determine the {\em a priori} probability for the presence of ``signal'' at each location-scale combination, and hence they tune the sparsity of the underlying signal, in terms of the proportion of locations-scale combinations with signals. In particular, they determine the quantities such as the prior probability for the null, that there is no factor contribution at all, and the prior expected number of location-scale combinations on which there is a factor contribution, etc. Depending on the inferential goal at hand, different strategies for specifying the sparsity hyparameters can be adopted. 

Specifically, if one's goal is for estimation and prediction, such as in signal denoising and image reconstruction, then a simple strategy is again to choose them by some data-fitting criteria such as empirical Bayes. Note that once a data-adaptive approach is taken in choosing the hyparameters, their face-values lose meanings. For example, one cannot interpret the ``prior'' null probabilities as before because it is chosen based on the data. 

In estimation and prediction one does not care so much about the prior interpretation of these parameters as finding the parameters that render the best predictive performance. Thus in such cases we can treat these sparsity hyparameters in the same way as the scaling parameters and choose them through empirical Bayes. If one's goal is hypothesis testing regarding factor effects, however, in order to maintain the ``validity'' of the test (such as its level from a frequentist perspective or the prior null probability from a Bayesian one), one should not use data to choose the sparsity parameters but should elicit them based on certain prior criteria. For example, one can choose $(\eta_{\kappa},\gamma_{\kappa})$ such that the prior joint alternative probability is say 50\%, as computed by applying Corollary~\ref{cor:pjap} to the prior model.  

In practical problems, one can use a hybrid strategy to specify the sparsity parameters. For example, if one is interested in testing the contribution of one or more factors, but not in the baseline, then one can specify $(\eta_{\kappa},\gamma_{\kappa})$ through prior elicitation, and use empirical Bayes to choose $(\eta_{\rho},\gamma_{\rho})$ which characterizes the baseline structure.
\vspace{0.5em}

{\em Empirical Bayes by maximum marginal likelihood.} We have referred multiple times to empirical Bayes as a strategy for choosing hyparameters. Specifically, a useful by-product of applying Theorems~\ref{thm:post_hmt} and \ref{thm:post_anova} is the overall marginal likelihood $P(\D\,|\,\bthe)=\xi_{0,0}(0)$, which can be computed through the pyramid algorithm. We can thus find the maximum marginal likelihood estimators (MMLE) for the hyperparameters $\hat{\bthe} = {\rm argmax}_{\bthe} P(\D\,|\,\bthe)$.
This optimization can be carried out using standard numerical methods. In our numerical studies, it is completed using the Nelder-Mead algorithm implemented in the {\tt R} function {\tt optim}. 

Alternatively we can use prior elicitation to choose $\bthe_{sparsity}$ and  MMLE to set the scaling parameters $\hat{\bthe}_{scaling}={\rm argmax}_{\bthe_{scaling}} P(\D\,|\,\bthe_{sparsity},\bthe_{scaling})$. We note that strictly speaking, the empirical Bayes strategy makes the uncertainty quantification such as credible intervals ``overly confident''. But because in the current context the total number of observations $nT$ is typically much much larger than the number of hyperparameters, this impact on the uncertainty quantification is often small.
\vspace{-1em}

\subsection{Decision rules and multiplicity control}
\label{sec:multiplicity}
\vspace{-0.5em}

Next we construct decision rules for calling ``significant'' factor effects. It is natural to reject the joint null hypothesis of no factor effects whatsoever when the PJAP is large for each factor of interest. But how large is large enough? The threshold for PJAP can either be chosen at a specific level such as 80\% provided that the prior, especially in terms of the sparsity parameters, is properly calibrated as described in the previous subsection. Alternatively, the PJAP threshold can be determined empirically by resampling strategies such as permutation.

In most applications, it is not only interesting to know whether a factor contributes to the variation at all but to understand the nature of such contribution. For example, what parts of the sample space is affected by that factor and at what scales. To this end, it is useful to consider decision rules for rejecting the node-specific hypotheses $\{H_{j,k}:(j,k)\in\T\}$ directly, and identify those locations-scale combinations $(j,k)$ at which the factor contributes to the variation. With this perspective, a natural decision rule for rejecting $H_{j,k}$ is when the PMAP $P(R_{j,k}=1\,|\,\D)>\delta$ for some threshold $\delta\in (0,1)$. Given any threshold $\delta$, the posterior expected {\em number of false positives} (NFP) regarding the $l$th factor, i.e., location-scale combinations that are called to be significant in the $l$th factor but for which the $l$th factor has no effect is given by
${\rm NFP}(\delta)=\sum_{(j,k):P(R_{j,k}=1\,|\,\D)>\delta} P(R_{j,k}=0\,|\,\D)$. Accordingly, the (Bayesian) false discovery rate (FDR) as defined in \cite{mueller_etal_2007} is given by 
\[
{\rm FDR}(\delta) = \frac{\sum_{(j,k):P(R_{j,k}=1\,|\,\D)>\delta} P(R_{j,k}=0\,|\,\D)}{|\{(j,k):P(R_{j,k}=1\,|\,\D)>\delta\}|}
\]
and thus one can choose $\delta$ to achieved the desired FDR.

\vspace{-1.5em}

\section{Numerical examples}
\label{sec:numerical_examples_fanova}
\vspace{-0.7em}

In this section we provide three numerical examples. In the first example we
apply the NIG-MT to denoising a single functional observation using the classical scenarios given in~\cite{donoho1994ideal}. In the second example we illustrate the work of NIG-MG for one-way fANOVA through simulation and compare it to existing wavelet-based fANOVA methods. Finally, we carry out two-way NIG-MG-based fANOVA for a well-known time-series data set, the ``orthosis'' data. In all the examples we use the Daubechies least-asymmetric orthonormal compactly
supported wavelet with 10 vanishing moments. 
\vspace{-1em}

\subsection{Function denoising}
\label{sec:denoising}
\vspace{-0.5em}

In this example we generate synthetic data from the four test functions proposed
by \cite{donoho1994ideal}, namely \emph{blocks}, \emph{bumps}, \emph{doppler}
and \emph{heavisine}. In \ref{fig:donoho_functions} we plot the four
functions and the associated mother wavelet coefficients. For each of the four test functions we consider four levels---1, 3, 5, and 7---of 
the root signal to noise ratio (RSNR):
${\rm RSNR} = \sqrt{ \sum_{t=1}^T (f_t - \bar{f})^2/(T-1)}/ \sigma$,
where $\bar{f} = \sum_t f_t/T$ and
$\Sigma_\epsilon = \text{diag}(\sigma^2, \ldots, \sigma^2)$. The observations
are taken at $T = 1024$ equidistant points, and for each function and each RSNR
level we generate 100 datasets.

In addition to our NIG-MT model, we apply two additional methods, one Bayesian and the other empirical Bayesian, for wavelet shrinkage---namely \cite{abramovich1998wavelet} and \cite{johnstone&silverman:2004}---to the simulated data as comparison. \cite{abramovich1998wavelet} is one of the early well-known Bayesian wavelet regression methods, while the two-group emprical Bayes method ``EBayesThresh'' introduced in \cite{johnstone&silverman:2004} is often regarded as the state-of-the-art in Bayesian wavelet denoising. We do not carry out a comprehensive study to the numerous available wavelet shrinkage methods, but include just these two related methods, especially since our main endeavor is for the fANOVA problem. For interested readers, a more extensive comparative study is available in \cite{johnstone&silverman:2004}, which showed that the {\tt EBayesThresh} approach compares favorably against existing methods. 

For each simulated functional observation we apply the three methods. We evaluate the performance of each method using the average (over location) mean square error (AMSE). For NIG-MT, we use empirical Bayes to set all of the hyperparameters. For the method of \cite{abramovich1998wavelet}, we set the two hyperparameters $\alpha=0.5$ and $\beta=1$ as recommended in that paper. The method is available in the R package {\tt wavethresh}. For the EBayesThresh method we use the Laplace prior for the wavelet coefficient under the alternative with the scale parameter of the Laplace set at $0.5$, which is the default value given in the {\tt R} package {\tt EBayesThresh}. Both {\tt wavethresh} and {\tt EBayesThresh} are available on {\tt CRAN}.

\ref{tab:function_denoising_amse} presents AMSE for each of the methods for all four signal functions at the four RSNR levels. The NIG-MT method outperforms the other methods in all simulation scenarios. The performance gain for doppler, bumps, and blocks are more substantial than that for heavisine, which is as expected because the wavelet coefficients for the first three functions display stronger spatial-scale dependency, and so exploiting such dependency is most rewarding.

\begin{table}[hp]
\doublespacing
\begin{center}
\begin{tabular}{lcccc}
\hline\hline
Doppler &\multicolumn{4}{c}{RSNR} \\
\cline{2-5}
&                          1     &   3   &     5  &        7         \\
NIG-MT     &  {\bf 5.2(0.7)}& \bf{1.1}(0.2)& \bf{0.47}(0.06)& \bf{0.27}(0.03)\\
EBayesThresh                &  6.2(0.9)& 1.2(0.2)& 0.52(0.08)& 0.32(0.03)\\
Bayes                   &  6.1(0.9)& 1.5(0.3)& 0.67(0.1) & 0.41(0.04)\\
\hline
Heavisine &\multicolumn{4}{c}{RSNR} \\
\cline{2-5}
&                          1     &   3   &     5  &        7         \\
NIG-MT    &   {\bf1.8(0.5)}& {\bf 0.51(0.1)} &  {\bf 0.26(0.04)}& {\bf 0.17(0.03)}\\
EBayesThresh                &  2.8(0.9)& 0.64(0.1)&  0.29(0.06)& 0.18(0.03)\\
Bayes                   &  2.3(0.8)& 0.71(0.1)&  0.4(0.06)&  0.26(0.04)\\

\hline
Bumps &\multicolumn{4}{c}{RSNR} \\
\cline{2-5}
&                          1     &   3   &     5  &        7         \\
NIG-MT     &  {\bf 17(2)}& {\bf 2.9(0.3)}& {\bf 1.3(0.1)}& {\bf 0.77(0.06)}\\
EBayesThresh                &  24(3)& 3.6(0.3)& 1.5(0.1)& 0.83(0.06)\\
Bayes                   &  30(3)& 5.5(0.5)& 2.9(0.2)& 2.2(0.2)  \\
\hline
Blocks &\multicolumn{4}{c}{RSNR} \\
\cline{2-5}
&                          1     &   3   &     5  &        7         \\

NIG-MT     &  {\bf 10(1)}& {\bf 2.1(0.2)}& {\bf 0.94(0.08)}& {\bf 0.55(0.05)}\\
EBayesThresh                &  12(1)& 3(0.2)&   1.3(0.1)&   0.73(0.07)\\
Bayes                   &  16(2)& 4.3(0.4)& 2.2(0.2)&   1.4(0.1)  \\
\hline\hline
\end{tabular}
\end{center}
\caption{AMSE ($\times 10^2$) of four methods---NIG-MT, EBayesThresh \cite{johnstone&silverman:2004} and Bayes \cite{abramovich1998wavelet} for four different signal functions at four RSNR levels, estimated from 100 simulations. Standard deviations of the AMSEs are given in parentheses. The lowest AMSE in each setting is in bold font.}
\label{tab:function_denoising_amse}
\end{table}

\vspace{-1em}

\subsection{Identifying functional variation}
\label{sec:simulations}
\vspace{-0.7em}

Next we carry out a simulation study to evaluate the performance of the NIG-MG method in fANOVA problems. In particular, we simulate from the one-way fANOVA model given in Eq.~\eqref{eq:oneway_model} and \eqref{eq:oneway_model2}. We consider the case with $G=3$ groups and each group has $n_g=3$ replicate observations. Following the notation in Eq.~\eqref{eq:oneway_model2}, we let $\vect{f}$ be the baseline mean for Group~1, and $\vect{b}^{(i)}$ be the contrast between Group~$i$ and Group~1 for $i=2,3$. For simplicity, we let $\vect{b}^{(2)}=-\vect{b}^{(3)}=\vect{b}$. That is, the difference between Group~2 and Group~1 is of the same magnitude as that of Group~3 and Group~1, but of the opposite sign. 

As in function denoising, the performance of a fANOVA method depends on the nature of the underlying function mean and factor effect. As such, we consider different scenarios in which $\vect{f}$ and $\vect{b}^{(i)}$ are of a variety of natures. The scenarios we consider fall into two broad categories. The first corresponds to the case when the cross-group difference is of a global nature, in the sense that they involve a large number of locations, whereas the second category corresponds to cases where the cross-group difference is local, involving only a very small subset of locations. We describe the different simulation scenarios in turn.

{\em Global factor effects.} In this case, we allow the baseline mean $\vect{f}$ and the factor difference $\vect{b}$ to be any of the four signature functions--doppler,heavisine,bumps,blocks, resulting in a total of 16 possible combinations. For each of the 16 combinations, we simulate 500 data sets, for RSNR=$0.5,1,1.5,2$. For each simulated data set, we also simulate a null data set resulting from setting $\vect{b}=\vect{0}$. We apply four methods---NIG-MG, the wfANOVA test~\cite{mckay:2013}, the tANOVA test~\cite{mckay:2013}, and the wavelet minimax test \cite{abramovich2004optimal}---to each of the simulated data set (both with a difference and the null set), and construct the ROC curve for the corresponding test statistic of each of the four methods. The wfANOVA is in essence the F-test applied to each location-scale combination. The tANOVA is the location-by-location F-test in the original space. The wavelet minimax test statistic is in essence the sum of squares of the wavelet coefficients after proper thresholding at the fine resolutions. For NIG-MG, we use the PJAP as the test statistic for the existence of a cross-group difference. 

\ref{fig:roc_cross_signal_rsnr1} and \ref{fig:roc_cross_signal_rsnr2} show the matrix of ROC curves for RSNR$=1$ and RSNR$=2$, respectively. (The ROC curves for RSNR$=0.5$ and RSNR$=1.5$ look similar and so we defer them to Supplementary Materials~S3. See \ref{fig:roc_cross_signal_rsnr0.5} and \ref{fig:roc_cross_signal_rsnr1.5}.) The NIG-MG outperforms the other methods by comfortable margins at all four RSNR levels. 
\begin{figure}[p]
 \centering
 \includegraphics[width= 1\textwidth]{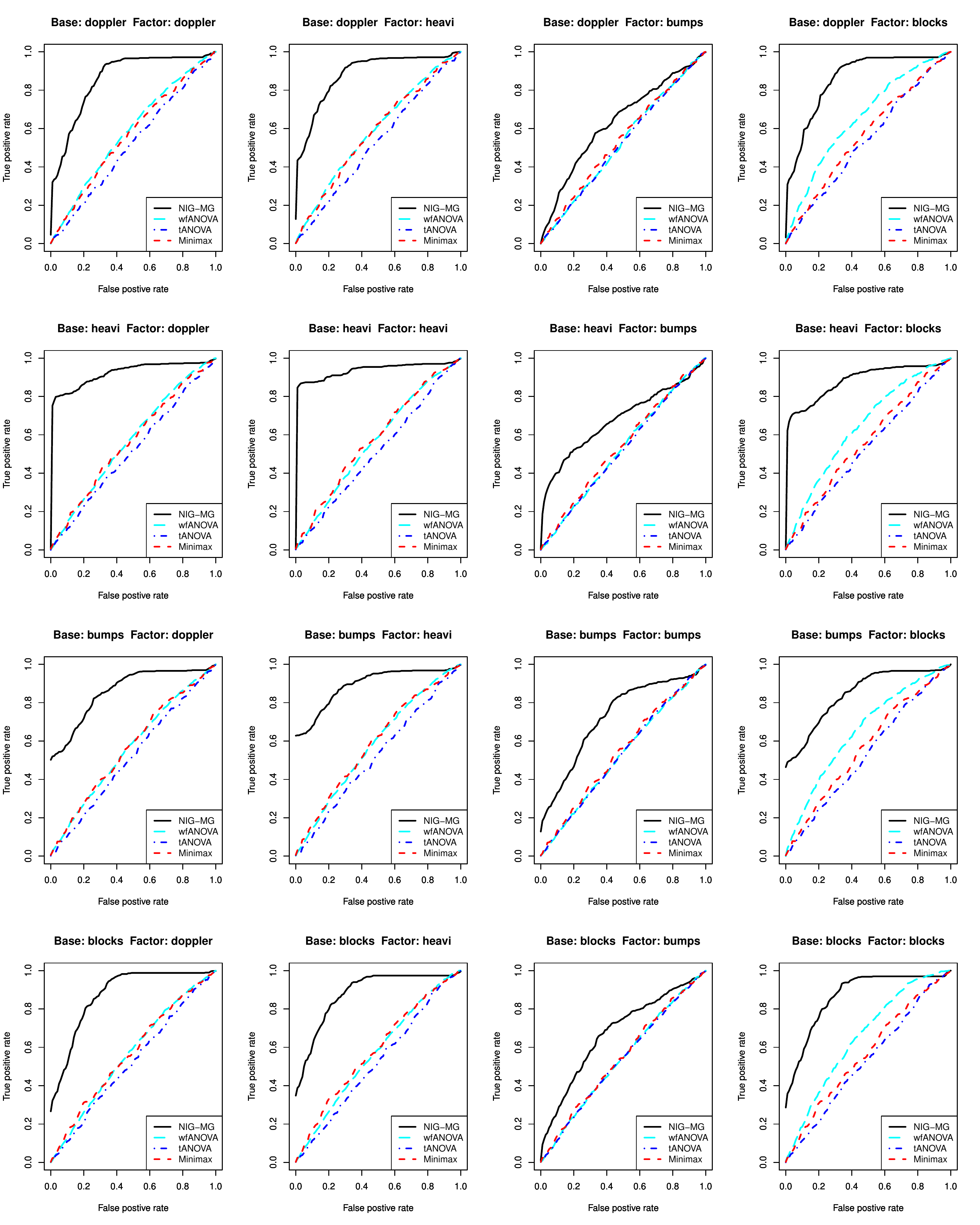}
 \caption{ROC curves (RSNR=$1$) for four test statistics when the baseline mean and the factor difference are all 16 combinations of the four signature functions.}
\label{fig:roc_cross_signal_rsnr1}
\end{figure}
\begin{figure}[p]
 \centering
 \includegraphics[width= 1\textwidth]{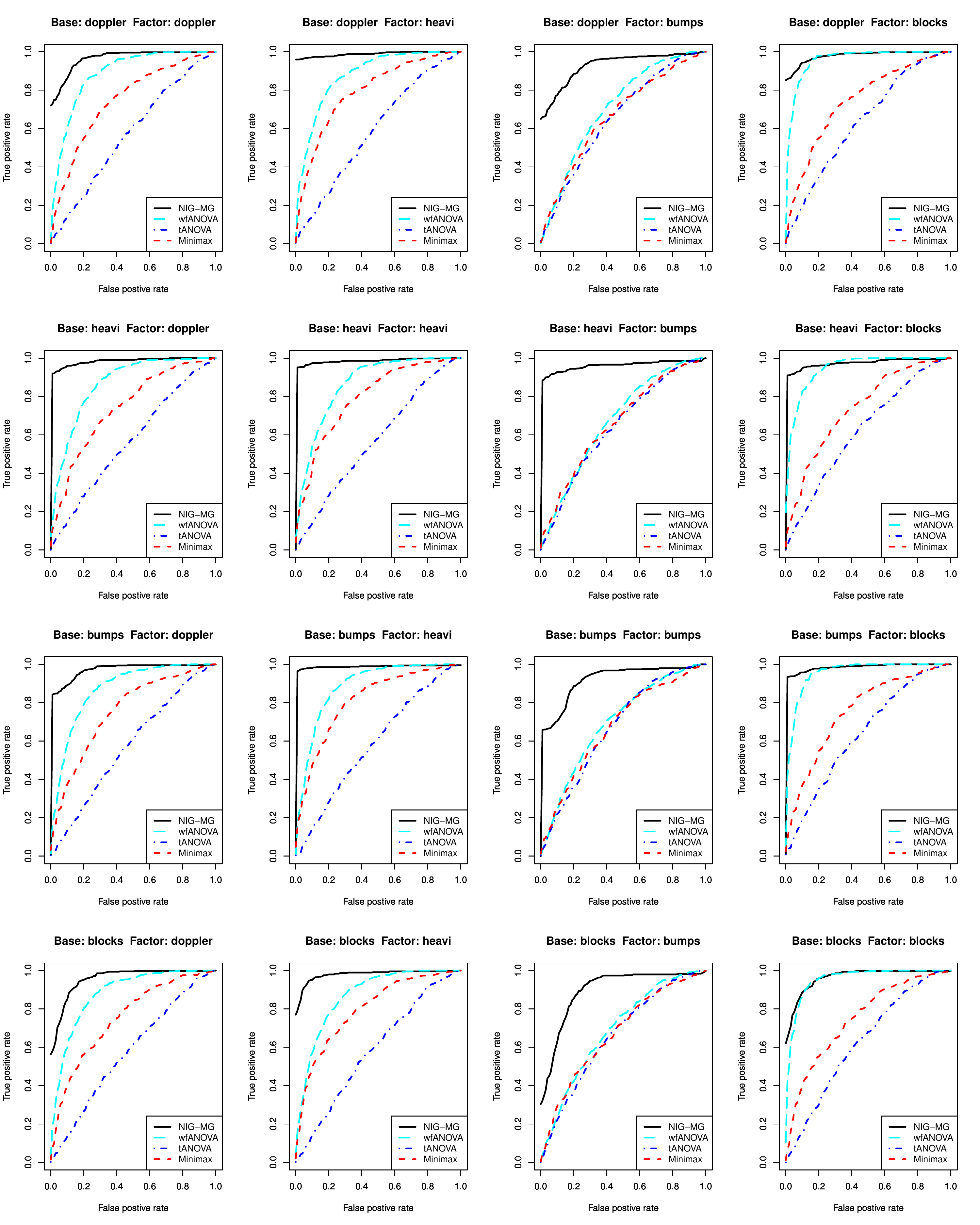}
 \caption{ROC curves (RSNR$=2$) for four test statistics when the baseline mean and the factor difference are all 16 combinations of the four signature functions.}
\label{fig:roc_cross_signal_rsnr2}
\end{figure}

{\em Local factor effects.} We then consider simulation scenarios in which the factor effects are of a local nature, involving only a small fraction of locations. We still let the baseline mean $\vect{f}$ to be any of the four signature functions. We let $\vect{b}$ be 0 for most locations, but for just a small interval, we let it be a constant proportion of $\vect{f}$. The mean functions for each of the three groups are plotted in the first column of  \ref{fig:roc_inject_signal}. 

We again apply the four methods to test for the existence of a cross-group difference. The second to fourth columns of \ref{fig:roc_inject_signal} show their ROC curves under three RSNR levels 2, 3, and 4. Again the performance advantage of NIG-MG is even more substantial than in the global difference scenarios. This shows the importance of incorporating the spatial-scale dependency when the underlying factor effects is local and so borrowing strength becomes critical to effectively identify such structures.
\begin{figure}[p]
 \centering
 \includegraphics[width= 0.98\textwidth]{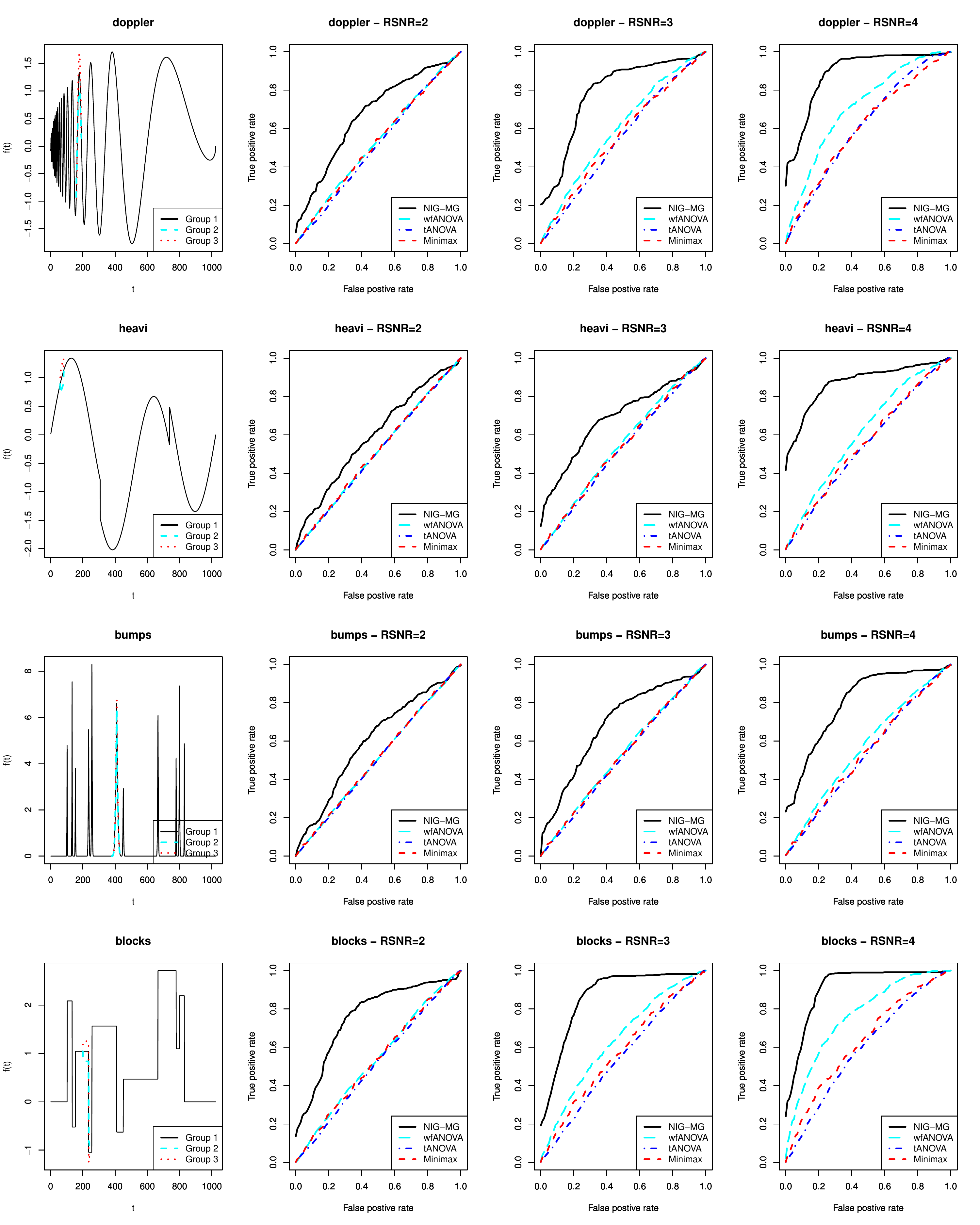}
 \vspace{-0.5em}
 \caption{ROC curves for four test statistics under the local cross-group difference scenarios. The rows correspond to different baseline mean functions---first: doppler; second: heavisine; third: bumps; four: blocks. The first columns shows the true mean functions of the three groups. The second to fourth columns correspond to RSNR=$2,3,4$ respectively.}
\label{fig:roc_inject_signal}
\end{figure}

\vspace{-1em}

\subsection{Orthosis Dataset}
\vspace{-0.5em}

We apply our NIG-MG model to analyze the
\emph{orthosis dataset}, a publicly available data original collected by Dr. David Amarantini and
Dr. Luc Martin from the Laboratoire Sport et Performance Motrice, Grenoble
University, France. This data has been used by several authors as a test-bed for functional data analysis methods \cite{abramovich2004optimal,abramovich2006testing,antoniadis2007estimation,zhang:2014}.  
The purpose of the study was to understand the effect of different types of
constraints to the knee on movement generation. In the study 7
individuals (i.e., the subjects) wore spring-loaded orthosis on the right knee while stepping in
place. Four experimental conditions were considered: a control
condition (without orthosis), an orthosis condition (with the orthosis only),
and two different springs loaded to the orthosis (spring 1 and spring 2). Ten replicated data sets were collected for each subject under each of the four conditions. The resultant moment for each trial was computed at $T=256$ equidistant time points. \ref{fig:orthosis_data} presents the entire data set. We refer the reader to \cite{cahouet2002static} for further detail on the experiment. 

\begin{figure}[p]
 \centering
 \includegraphics[width= 1\textwidth]{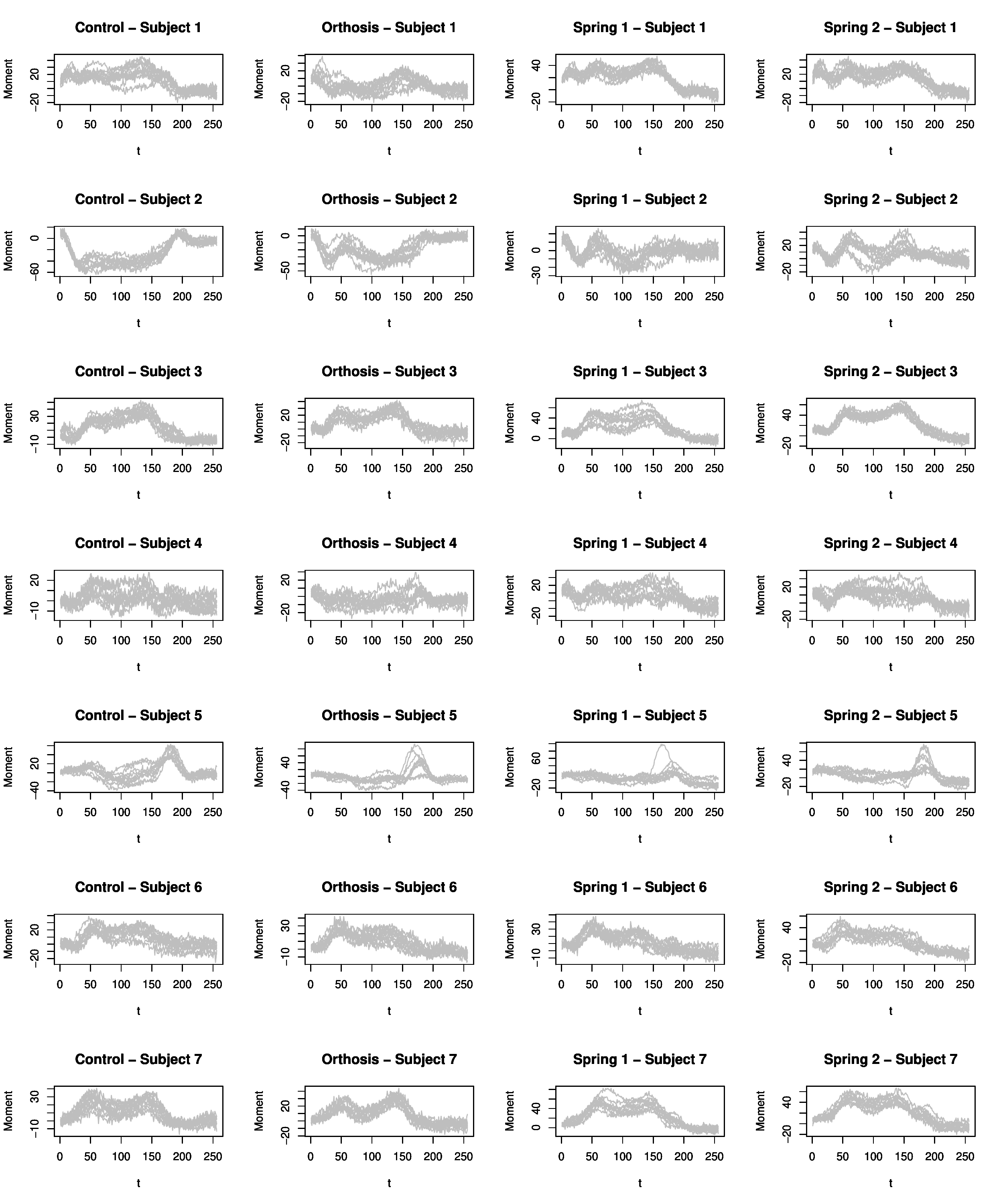}
 \caption{The orthosis data. Each row
corresponds to a subject, each column to an experimental condition, and each curve is a replicate. }
\label{fig:orthosis_data}
\end{figure}

A key question to address from this data is how the four experimental conditions result in different knee movement as measured by the functional shape of the measured moments. 
While the subject-to-subject variation is not of direct interest, it is substantial and must be properly taken into account. Treating both the experimental conditions and the subjects as factors, the experiment corresponds to a two-way ANOVA design. We apply our NIG-MG model for two-way fANOVA. We set the prior null probability to about 50\% with $\eta=0.3$ and $\gamma=0.4$. The posterior probability for the joint null that there is no difference among the four conditions is virtually zero. \ref{fig:orthosis_pmaps} presents the PMAPs for the two factors.
\begin{figure}[t]
 \centering
 \includegraphics[width= 0.95\textwidth]{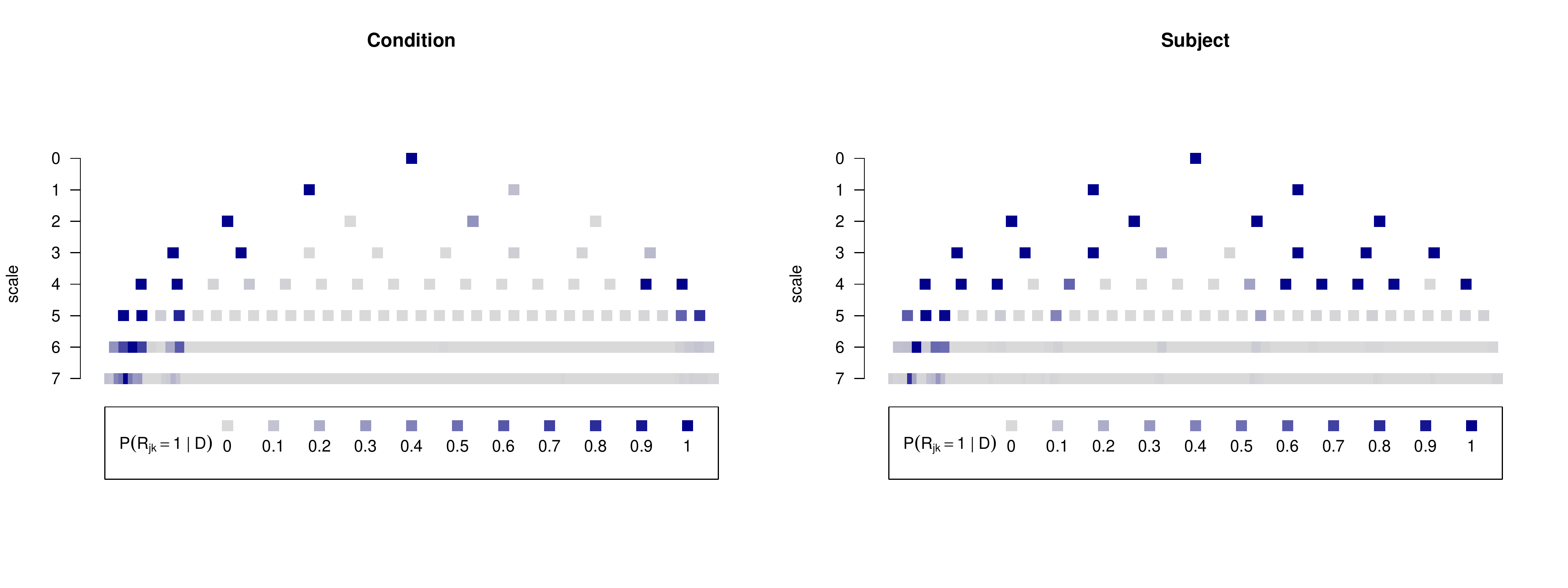}
\vspace{-3em}

 \caption{The PMAPs for both factors---the experiment condition (left) and the subject (right) under the NIG-MG model.}
\label{fig:orthosis_pmaps}
\end{figure}

\begin{figure}[p]
 \centering
 \includegraphics[width= 1\textwidth]{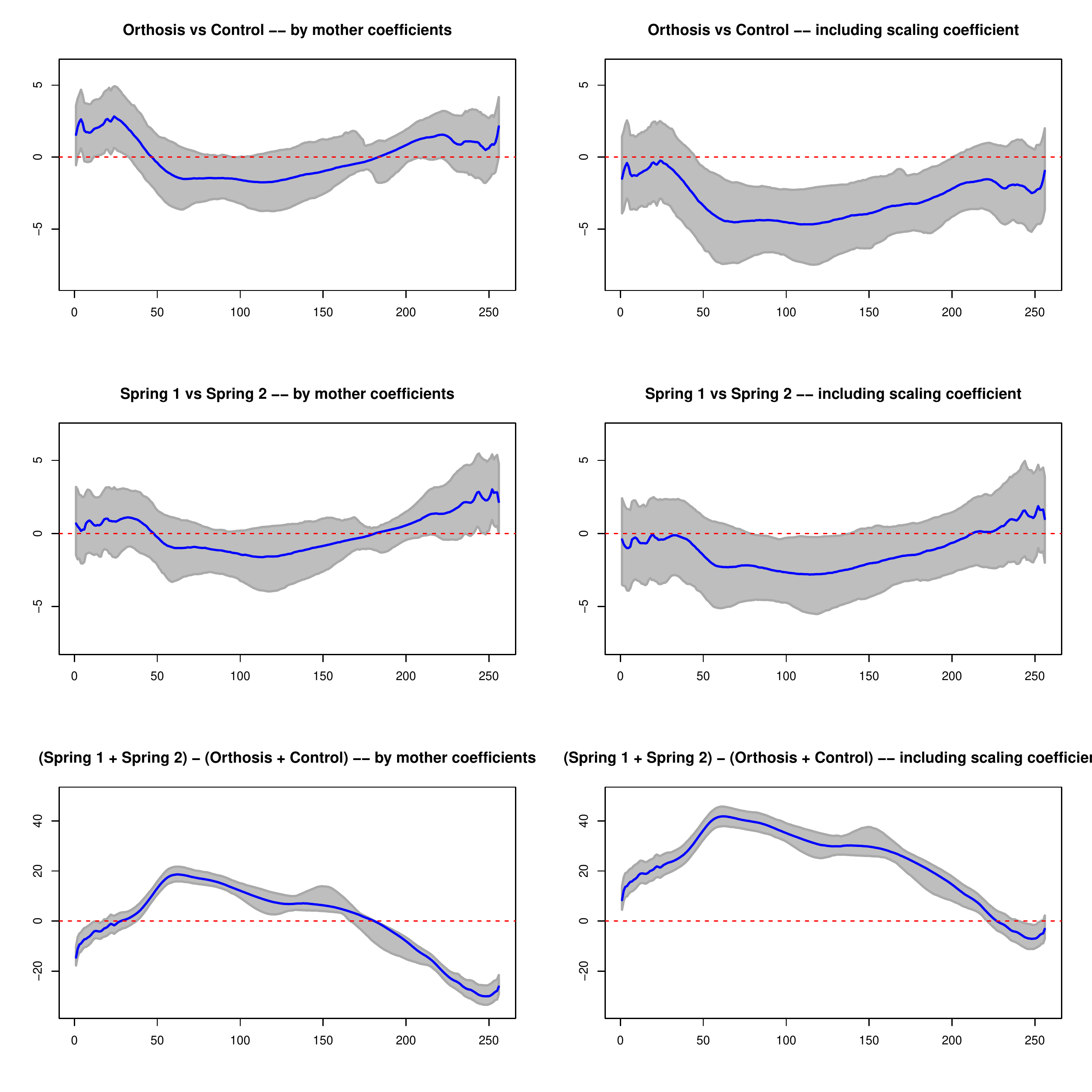}
\vspace{-1em}

 \caption{Pointwise credible bands for three different contrasts---Row 1: Orthosis versus Control, Row 2: Spring 1 vs Spring 2, Row 3: Spring vs No spring. Left column: Credible bands constructed based on the posterior samples of the mother wavelet coefficients from the NIG-MG model. Right column: Credible bands constructed based on the posterior samples of the mother wavelet coefficients fro the NIG-MG model as well as posterior draws for the scaling (father wavelet) coefficient.}
\label{fig:orthosis_credible_bands}
\end{figure}

While there is strong evidence for difference across the experiment conditions (as well as the subjects), one may also be interested in investigating the factor contributions with regard to certain contrasts---e.g., the difference between the orthosis only condition and the control, the difference between the two different springs, and the difference between the spring vs no spring conditions. The fully probabilistic nature of the NIG-MG framework allow us to address such tasks while properly taking into account the uncertainty involved through the standard Bayesian recipe---sampling from the joint posterior and construct credible bands for the corresponding contrasts. 

\ref{fig:orthosis_credible_bands} presents the credible bands for these three contrasts. In particular, for each contrast we construct two credible bands. The first credible band (shown in the left column in the figure) is constructed using the posterior samples for just the mother wavelet coefficients, excluding the scale coefficient. The quantify the uncertainty in the shape of the functional contrast but not in the mean level of the contrast across all locations (here the time points). In addition, we create also the credible band that uses the posterior samples for both the mother and the scaling coefficients, which incorporates both the functional shape and the mean while taking into account the uncertainty from both sources. According to our knowledge, all of the previous analysis of the data set using wavelet-based methods \cite{abramovich2004optimal,abramovich2006testing,antoniadis2007estimation} only provides point estimate of the contrasts without providing uncertainty quantification.

 \vspace{-1.5em}
 
\section{Conclusion}
\label{sec:conclusion}
\vspace{-0.7em}

We have introduced a new Bayesian hierarchical model in the wavelet domain for addressing the functional analysis of variance problem.  By incorporating a graphical model that links the presence and absence of factor effects on the wavelet coefficients, this model allos effective borrowing of information across locations and scales, and this in turn leads substantial performance gain over methods that ignore such dependency, especially in situations where the underlying factor effects are local. Moreover, the exact posterior of the model can be computed exactly through an efficient pyramid type recursive algorithm which is linear in both the number of observations and the number of locations, i.e., the complexity as Mallat's pyramid algorithm for DWT. In addition, the fully probabilistic nature of the model allows inference to be carried out in a principled manner---uncertainty quantification is achieved through posterior probabilities and credible bands.

The computational complexity of the pyramid algorithm for evaluating the exact posterior of NIG-MG is linear in both $n$ and $T$. So inference scales well with the number of functional observations $n$ as well as the number of locations $m$. However, as the number of factors $L$ grows, the computation scales as $O(4^L)$ and so can become infeasible if the number of factors $L$ is large. In most applications of fANOVA, however, the number of factors $L$ is usually small $\leq 5$ for which exact inference can be completed quickly.

\bibliography{References}{}

\newpage

\beginsupplement

\section*{Supplementary Materials}
\subsubsection*{S1.~Proofs}
\begin{proof}[Proof of Theorem~\ref{thm:post_hmt}]
This theorem is a special case of Theorem~\ref{thm:post_anova} and thus follows from Theorem~\ref{thm:post_anova}.
\end{proof}

\begin{proof}[Proof of Theorem~\ref{thm:post_anova}]
First define $\D_{j,k}$ to be the data, i.e., the empirical wavelet coefficients observed on all subtree rooted at node $(j,k)$, $\T(j,k)$. Then note that $\xi(s,r)$ is the marginal likelihood on the subtree $\T(j,k)$ given the event that $S_{j-1,\lfloor k/2, \rfloor}=s$ and $R_{j-1,\lfloor k/2, \rfloor}=r$, and $\phi(s,r)$ is the marginal likelihood on $\T(j,k)$ given the event that $S_{j,k}=s$ and $R_{j,k}=r$.

Now by Bayes theorem, 
\begin{align*}
&\Pr(S_{j,k} =  s', R_{j,k} =  r' | S_{j-1,\lfloor k/2 \rfloor} = s, 
R_{j-1,\lfloor k/2 \rfloor}=
r,\D)\\
=&\frac{\Pr(S_{j,k} =  s', R_{j,k} =  r', \mathcal{D}_{j,k} |  S_{j-1,\lfloor k/2 \rfloor} = s, R_{j-1,\lfloor k/2 \rfloor}=r)}{\Pr(\mathcal{D}_{j,k} |  S_{j-1,\lfloor k/2 \rfloor} = s, R_{j-1,\lfloor k/2 \rfloor}=r)}\\ 
=& \frac{\Pr(S_{j,k} =  s', R_{j,k} =  r' |  S_{j-1,\lfloor k/2 \rfloor} = s, R_{j-1,\lfloor k/2 \rfloor}=r)\cdot \Pr(\mathcal{D}_{j,k} |  S_{j,k} = s', R_{j,k}=r')}{\Pr(\mathcal{D}_{j,k} |  S_{j-1,\lfloor k/2 \rfloor} = s, R_{j-1,\lfloor k/2 \rfloor}=r)}\\
=&\rho_{j,k}(s,s') \kappa_{j,k}(r,r') 
\phi_{j,k}(s',r')/\xi_{j,k}(s,r)
\end{align*}
for $j\geq 1$. Now for $j=0$, similarly, 
\begin{align*}
&\Pr(S_{j,k} =  s', R_{j,k} =  r' | \D)\\
=&\Pr(S_{j,k} =  s', R_{j,k} =  r', \mathcal{D}_{j,k})/\Pr(\mathcal{D}_{j,k})\\ 
=& \Pr(S_{j,k} =  s', R_{j,k} =  r' )\cdot \Pr(\mathcal{D}_{j,k} |  S_{j,k} = s', R_{j,k}=r')/\Pr(\mathcal{D}_{j,k})\\
=&\rho_{0,0}(s') \kappa_{0,0}(r') 
\phi_{0,0}(s',r')/\xi_{0,0}(0,0).
\end{align*}
Now we show that the marginal likelihood $\xi$ and $\phi$ indeed follow the recursive expression. To this end, note that by definition, for $j=J$,
\begin{align*}
\phi_{j,k}(s,r) = \Pr(\D_{j,k}|S_{j,k}=s,R_{j,k}=r)=m(s,r).
\end{align*}
Then also by definition for $j>0$,
\begin{align*}
\xi_{j,k}(s,r) &= \Pr(\D_{j,k} |  S_{j-1,\lfloor k/2 \rfloor} = s, R_{j-1,\lfloor k/2 \rfloor}=r)\\
&= \sum_{s',r'}  \Pr(S_{j,k}=s',R_{j,k}=r' |  S_{j-1,\lfloor k/2 \rfloor} = s, R_{j-1,\lfloor k/2 \rfloor}=r)\cdot \Pr(\D_{j,k}|S_{j,k}=s,R_{j,k}=r)\\
&=\sum_{s',r'} \rho_{j,k}(s,s') \cdot \kappa_{j,k}(r,r') \phi_{j,k}(s',r')
\end{align*}
and for $j=0$,
\begin{align*}
\xi_{j,k}(s,r) &= \Pr(\D_{j,k})= \sum_{s',r'}  \Pr(S_{j,k}=s',R_{j,k}=r' )\cdot \Pr(\D_{j,k}|S_{j,k}=s,R_{j,k}=r)\\
&=\sum_{s',r'} \rho_{0,0}(s') \cdot \kappa_{0,0}(r') \phi_{0,0}(s',r')
\end{align*}
Similarly, for $j<J$, 
\begin{align*}
\phi_{j,k}(s,r) &= \Pr(\D_{j,k}|S_{j,k}=s,R_{j,k}=r)\\
&=m(s,r)\cdot \Pr(\D_{j+1,2k},\D_{j+1,2k+1}|S_{j,k}=s,R_{j,k}=r)\\
&=m_{j,k}(s,r) \cdot \xi_{j+1,2k}(s,r) \cdot \xi_{j+1,2k+1}(s,r).
\end{align*}
This shows that the marginal posterior on the latent variables is a MT with states $\{0,1\}\times\{0,1\}$ the claimed transition matrix. 
The conditional posteriors of the regression coefficients and the errors follow directly from standard results on Bayesian linear regression with the NIG conjugate prior. 
\end{proof}

\subsubsection*{S2.~Posterior NIG-MG for multiple factors}
With $L$ factors, the design matrix is now
\vspace{-1.7em}

 $$
X = ( \mathbbm{1}_n, \vect{e}^{(1)}_2, \ldots, \vect{e}^{(1)}_{G_1},\ldots, \vect{e}^{(L)}_2, \ldots, \vect{e}^{(L)}_{G_L}),
$$
\vspace{-2.5em}

\noindent where $ n = \sum_{l=1}^{L}\sum_{g=1}^{G_l} n_g$, $\mathbbm{1}_n$ is a vector of $n$ ones, and $\vect{e}^{(l)}_g$ is a binary vector where the
$h$th element is equal to one if  the $h$th observation belongs to
the $g$th group for the $l$th factor, and equal to zero otherwise. We also define the following matrices for $(s,\vect{r})=(s,r_1,r_2,\ldots,r_L)\in\{0,1\}^{L+1}$:
\vspace{-1.5em}

\[
X(s,\vect{r}) = 
(s\mathbbm{1}_n, r_1\vect{e}^{(1)}_2, \ldots, r_1\vect{e}^{(1)}_G,\ldots,r_L\vect{e}^{(L)}_2, \ldots, r_L\vect{e}^{(L)}_G ),\]
\[ 
\Lambda_{j} = 
\text{diag}(1/\tau_j,\underbrace{1/\upsilon^{(1)}_j,\ldots,1/\upsilon^{(1)}_j}_{\text{$G_{1}-1$ copies}},\ldots, \underbrace{1/\upsilon^{(L)}_j,1/\upsilon^{(L)}_j,\ldots,1/\upsilon^{(L)}_j}_{\text{$G_{L}-1$ copies}})
\]
where $\upsilon^{(l)}_{j}=2^{-\alpha j} \upsilon_{l}$ and 
\[
M(s,\vect{r}) = 
 \left(
\begin{array}{c}
 s \\
 r_1 \mathbbm{1}_{G_{1}-1} \\
 r_2 \mathbbm{1}_{G_{2}-1} \\
 \vdots\\
 r_L \mathbbm{1}_{G_{L}-1}
\end{array} 
\right)
\left(
\begin{array}{ccccc}
 s & r_1 \mathbbm{1}_{G_{1}-1}' & r_2 \mathbbm{1}_{G_{2}-1}' & \cdots & r_L \mathbbm{1}_{G_{L}-1}'
\end{array} 
\right)
.
\]

\noindent The marginal likelihood for the node-specific model on
$(j,k)$ given $S_{j,k} = s$ and $\vect{R}_{j,k} = \vect{r}$ is
\begin{equation}\label{eq:marginal_like}
m_{j,k}(s,\vect{r}) = 
  \frac{( \nu \sigma_0^2 )^{\nu+1} \Gamma(\nu + n/2+1)}{(2\pi)^{n/2}\Gamma(\nu + 1)} \cdot \dfrac{ |\Lambda_{j}|^{1/2} }{|\Lambda_{j}^{*}(s,\vect{r})|^{1/2}} 
  \cdot \big[ \nu \sigma_0^2
+ \Upsilon_{j,k}(s,\vect{r}) \big]^{-\nu-n/2-1}.
\end{equation}
where
\vspace{-3.5em}

\begin{align*}
&\Upsilon_{j,k}(s,\vect{r})=\left\{ \vect{d}_{j,k}'\vect{d}_{j,k} 
  - [\vect{\mu}_{j,k}^*(s,\vect{r})]'   \Lambda_{j}^{*}(s,\vect{r})\vect{\mu}_{j,k}^*(s,\vect{r})\right\}/2,\\
\Lambda_{j}^*(s,\vect{r}) &= X(s,\vect{r})'X(s,\vect{r}) + \Lambda_{j}, \quad \text{and} \quad \vect{\mu}_{j,k}^*(s,\vect{r}) = [\Lambda_{j}^*(s,\vect{r})]^{-1} [
 X(s,\vect{r})' \vect{d}_{j,k}  ].
 \end{align*}
\vspace{-2.5em}

\begin{theorem}
\label{thm:post_anova_multi_factor}
 The joint posterior on  $\{z_{j,k},
\bbet_{j,k}, S_{j,k}, \vect{R}_{j,k}, \sigma_{j,k}^2 : (j,k)
\in \mathcal{T} \}$ under the $L$-factor NIG-MG is as follows.
\begin{itemize}
 \item The marginal posterior of the hidden states 
$ \{ (S_{j,k}, \vect{R}_{j,k}) : (j,k)
\in \mathcal{T} \}$ is an MT defined on the product state-space $\{0,1\}^{L+1}$ with 
\begin{enumerate}
 \item State transition probabilities:
 \vspace{-3.5em}
 
\begin{align*}
\hspace{-4em} \Pr(S_{j,k} =  s', \vect{R}_{j,k} =  \vect{r}' | S_{j-1,\lfloor k/2 \rfloor} = s, 
\vect{R}_{j-1,\lfloor k/2 \rfloor} =
\vect{r},\mathcal{D})= \rho_{j,k}(s,s') \prod_{l=1}^{L} \kappa_{j,k}(r_l,r_l') \cdot
\frac{\phi_{j,k}(s',\vect{r}')}{\xi_{j,k}(s,\vect{r})},
\end{align*}
\vspace{-3.5em}
 
\noindent for $j=1, \ldots, J$.
 \item Initial state probabilities:
 \vspace{-2em}
 
$$
\Pr(S_{0,0}=s, \vect{R}_{0,0}=\vect{r} | \mathcal{D}) = 
\rho_{0,0}(s) \prod_{l=1}^{L} \kappa_{0,0}(r_l) \cdot
\frac{\phi_{0,0}(s,\vect{r})}{\xi_{0,0}(0,\vect{0})}.
$$
\end{enumerate}
\vspace{-0.5em}

\item The conditional posterior of $\sigma_{j,k}^2$ given $S_{j,k}$ and
$\vect{R}_{j,k}$ is: 
\vspace{-2em}
 
$$
 [ \sigma_{j,k}^2 | S_{j,k}, \vect{R}_{j,k}, \mathcal{D}] 
\sim \text{Inv-Gamma}\bigg( \nu + 1 + \dfrac{n}{2}, 
\nu \sigma_0^2
+ \Upsilon_{j,k}(S_{j,k},\vect{R}_{j,k})    \bigg).
$$

\item The posterior of $z_{j,k},
\beta_{1\, j,k}^{(2)}, \ldots, \beta_{1\,j,k}^{(G_1)},\ldots,\beta_{L\, j,k}^{(2)}, \ldots, \beta_{L\,j,k}^{(G_L)}$ given $S_{j,k}$, $\vect{R}_{j,k}$ and
$\sigma_{j,k}^2$ is given as follows 
\vspace{-3.5em}

\begin{align*}
&[ z_{j,k}, \beta_{1\, j,k}^{(2)}, \ldots, \beta_{1\,j,k}^{(G_1)},\ldots,\beta_{L\, j,k}^{(2)}, \ldots, \beta_{L\,j,k}^{(G_L)}\, |\, 
  \sigma_{j,k}^2, S_{j,k}, \vect{R}_{j,k},
 \mathcal{D}]\\ 
& \hspace{5em} \sim {\rm N}\bigg( \vect{\mu}^*_{j,k}(S_{j,k},\vect{R}_{j,k}), \sigma_{j,k}^2\,M(S_{j,k},\vect{R}_{j,k}) \circ  [
 \Lambda_{j}^*(S_{j,k},\vect{R}_{j,k}) ]^{-1} \bigg).
\end{align*}
\vspace{-3.5em}
\end{itemize}
The mappings $\phi_{j,k}(s,\vect{r})$ and $\xi_{j,k}(s,\vect{r}): \{0,1 \}^{L+1}
\to [0, + \infty) $ can be computed recursively through a pyramid algorithm as follows:
\vspace{-2.5em}

\begin{align*}
\phi_{j,k}(s,\vect{r})& =
\left\{ 
\begin{array}{ll}
 m_{j,k}(s,\vect{r}) \cdot \xi_{j+1,2k}(s,\vect{r}) \cdot \xi_{j+1,2k+1}(s,\vect{r}) & \text{for
$j=0,1,2\ldots,J-1$}\\
m_{j,k}(s,\vect{r}) & \text{for $j=J$,}
\end{array}  \right. \\
 \xi_{j,k}(s,\vect{r}) & = 
\left\{
\begin{array}{ll}
 \sum_{ s', \vect{r}' }
\rho_{j,k}(s,s') \cdot \prod_{l=1}^{L}\kappa_{j,k}(r_l,r_l') \cdot \phi_{j,k}(s',\vect{r}')  & \text{for
$j=1,2,\ldots,J$}\\
\sum_{s', \vect{r}'} \rho_{0,0}(s') \cdot \prod_{l=1}^{L}\kappa_{0,0}(r_l) \cdot
\phi_{0,0}(s',\vect{r}') & \text{for $j=0$.}
\end{array}
\right.
\end{align*}
\end{theorem}
\newpage 

\subsubsection*{S3.~Additional figures}
\vspace{-1em}

\begin{figure}[h!]
\vspace{-1em}

 \centering
 \includegraphics[width= 0.97\textwidth]{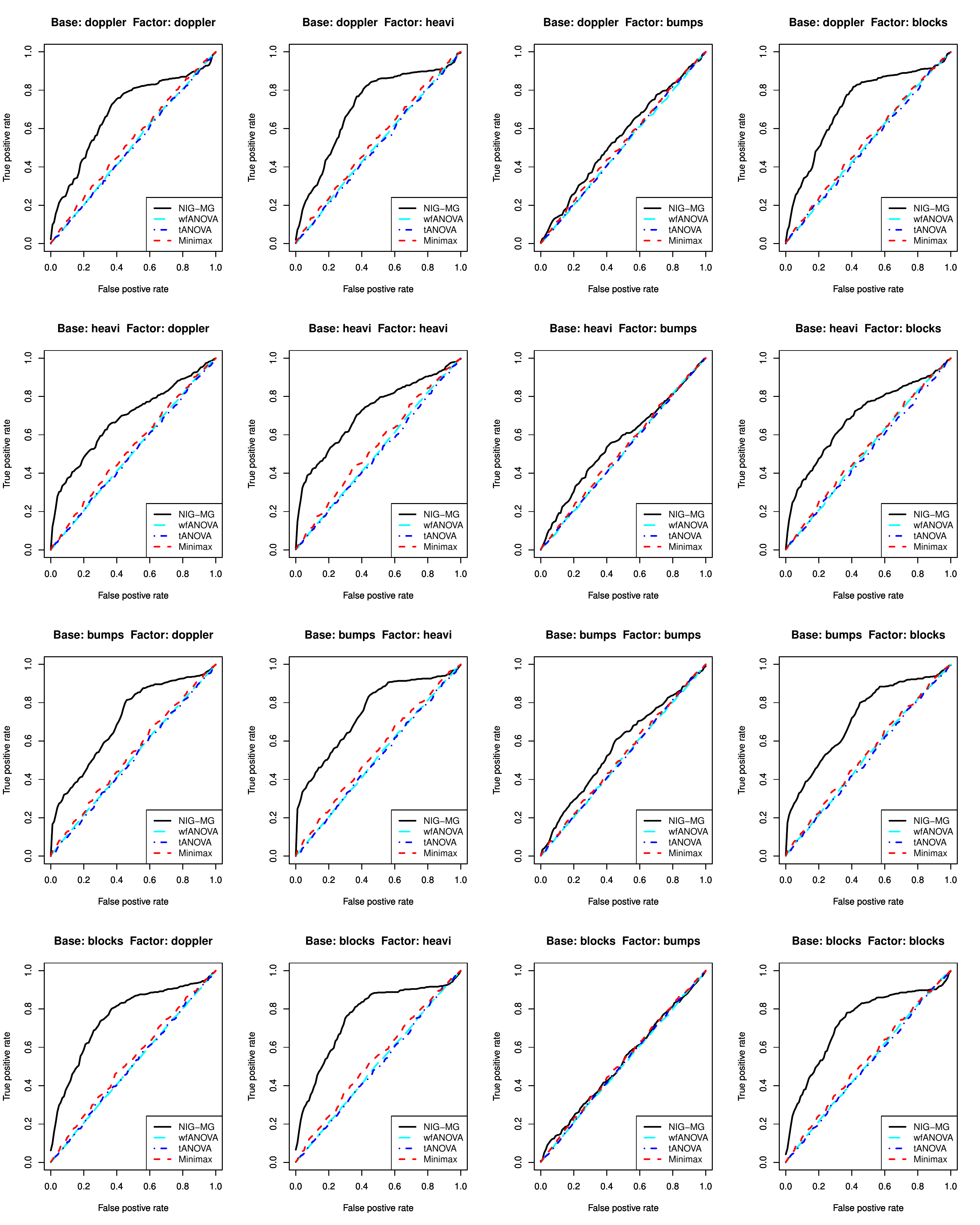}
 \vspace{-0.5em}
 
 \caption{ROC curves (RSNR=$0.5$) for four test statistics when the baseline mean and the factor difference are all 16 combinations of the four signature functions.}
\label{fig:roc_cross_signal_rsnr0.5}
\end{figure}
\begin{figure}[p]
 \centering
 \includegraphics[width= 1\textwidth]{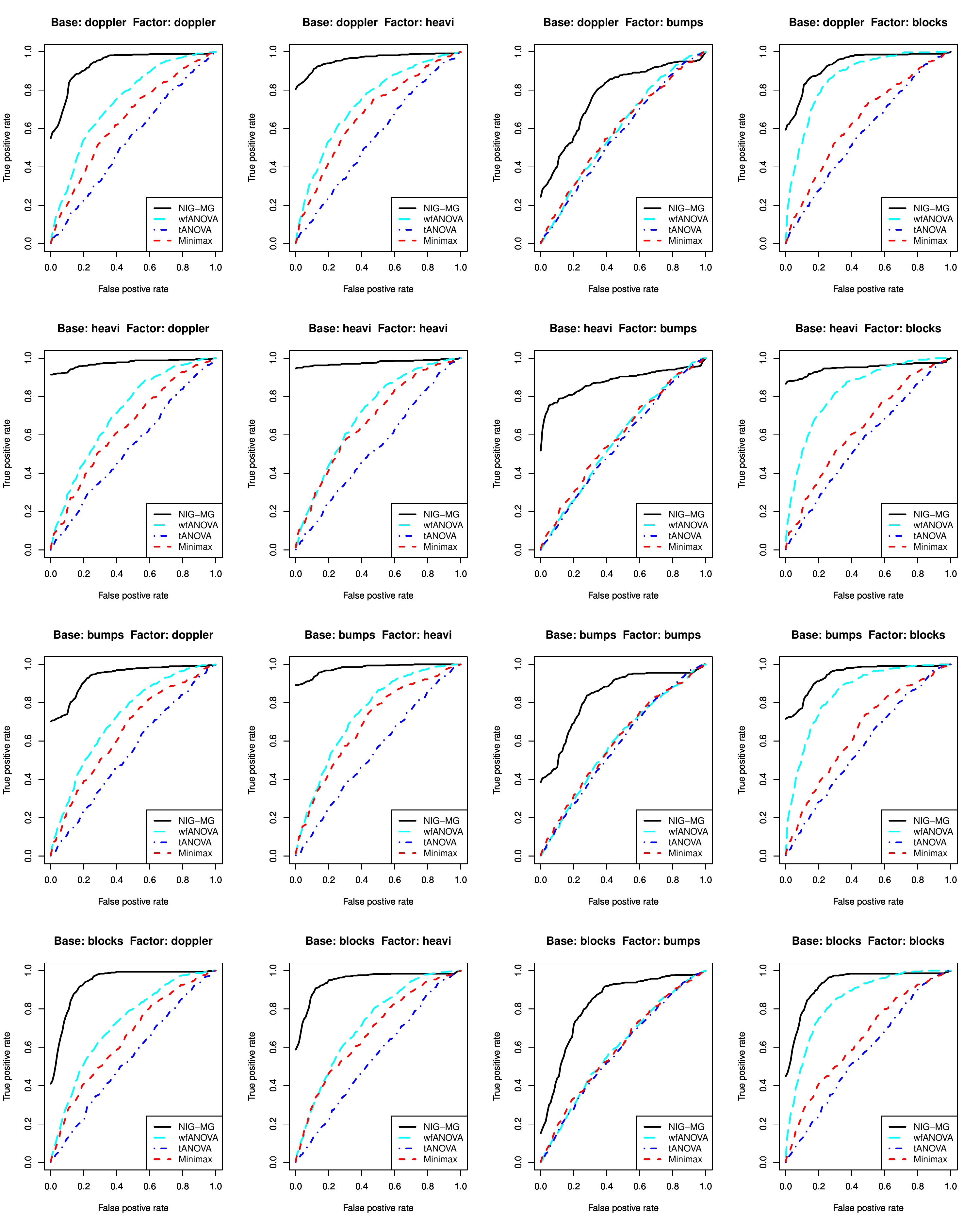}
 \caption{ROC curves (RSNR$=1.5$) for four test statistics when the baseline mean and the factor difference are all 16 combinations of the four signature functions.}
\label{fig:roc_cross_signal_rsnr1.5}
\end{figure}

\end{document}